\documentclass[14pt]{article}
\usepackage{hyperref}
\usepackage{amssymb}
\usepackage{amsmath}
\usepackage{mdframed}
\usepackage{subfig}
\usepackage{cancel}
\usepackage{enumerate}
\usepackage{color}
\usepackage{mathtools}
\usepackage{pdfpages}
\usepackage{mathrsfs}
\usepackage{enumitem}
\usepackage{mathtools}
\usepackage{graphicx}
\usepackage{mdframed}
\usepackage{amsfonts}
\usepackage{cancel}
\usepackage{placeins}
\usepackage{amsthm}
\usepackage{tikz}
\usetikzlibrary{shapes,arrows,automata}
\usepackage{xcolor}
\usetikzlibrary{arrows.meta}
\usepackage{makecell}

\definecolor{goldenpoppy}{rgb}{0.99, 0.76, 0.0}
\definecolor{richblack}{rgb}{0.06, 0.05, 0.03}
\definecolor{cadmiumred}{rgb}{0.89, 0.0, 0.13}
\definecolor{fuchsia}{rgb}{0.3, 0.0, 0.3}
\definecolor{green(ncs)}{rgb}{0.0, 0.52, 0.32}
\tikzstyle{species_T} = [circle,radius=0.1cm, text centered, draw=black, fill=goldenpoppy]
\tikzstyle{species_R} = [circle,radius=0.1cm, text centered, draw=black, fill=white]
\tikzstyle{species_C} = [circle,radius=0.1cm, text centered, draw=black, fill=fuchsia]
\tikzstyle{dots} = [circle,radius=0.1cm,text centered]
\tikzstyle{arrowA} = [-{Latex[length=2mm]},white,dashed]
\tikzstyle{arrowB} = [-{Latex[length=2mm]},green(ncs)]
\tikzstyle{arrowC} = [-{Latex[length=2mm]},cadmiumred]
\tikzstyle{inhibit} = [thick,-|,black!100]
\tikzstyle{loosely dashed}= [dash pattern=on 3pt off 6pt]

\newtheorem{definition}{Definition}[section]
\newtheorem{theorem}{Theorem}[section]

\newtheorem{example}{Example}[section]

\newtheorem{problem}{Problem}[section]

\let\Item\item

\mdfdefinestyle{theoremstyle}{%
frametitlerule=true,%
innertopmargin=\topskip,
}
\mdtheorem[style=theoremstyle]{algorithm}{Algorithm}[section]

\begin{document}
\vspace*{-1cm}

\centerline{{\huge Simple chemical systems with chaos}}

\medskip
\bigskip

\centerline{
\renewcommand{\thefootnote}{$1$}
{\Large Tomislav Plesa \footnote{
Department of Applied Mathematics and Theoretical Physics, University of Cambridge,
Centre for Mathematical Sciences, Wilberforce Road, Cambridge, CB3 0WA, UK;
e-mail: tp525@cam.ac.uk}, \renewcommand{\thefootnote}{$2$}
Julien Clinton Sprott \footnote{
Department of Physics, University of Wisconsin,
Madison, WI 53706, USA}
}}

\medskip
\bigskip

\noindent
{\bf Abstract}: A number of simple chaotic three-dimensional 
dynamical systems (DSs) with quadratic polynomials 
on the right-hand sides 
are reported in the literature, 
containing exactly $5$ or $6$
monomials of which only 
$1$ or $2$ are quadratic.
However, none of these simple systems are
chemical dynamical systems (CDSs)
- a special subset of polynomial DSs 
that model the dynamics
of mass-action chemical reaction networks (CRNs). 
In particular, only a small number of
three-dimensional quadratic CDSs with chaos 
are reported, 
all of which have at least $9$ monomials 
and at least $3$ quadratics, 
with CRNs containing at least $7$ reactions 
and at least $3$ quadratic ones.
To bridge this gap, in this paper 
we prove some basic properties of chaotic CDSs, 
including that those in three dimensions
have at least $6$ monomials,
at least one of which is  
negative and quadratic. 
We then use these results 
to computationally find $20$ chaotic three-dimensional 
CDSs with $6$ monomials and as few as $4$ quadratics, 
or $7$ monomials and as few as $2$ quadratics.
At the CRN level, some of these systems
have $4$ reactions of which only $3$ are quadratic, 
or $5$ reactions with only $2$ being quadratic.
These results quantify structural
complexity of chaotic CDSs, and indicate 
that they are ubiquitous. 

\section{Introduction}
Systems of $N$ first-order autonomous
ordinary-differential equations with 
polynomials of at most degree $n$
on the right-hand side, called
$N$-dimensional $n$-degree 
polynomial \emph{dynamical systems} (DSs), 
can display chaos when $N \ge 3$
and $n \ge 2$~\cite{Wiggins,Sprott_book_1}. 
A fundamental task is to characterize 
structurally the simplest 
three-dimensional quadratic DSs that display
this complicated dynamical behavior.
To this end, we associate 
to each such DS a vector $(a,b)$, 
where $a$ is the total
number of monomials on the right-hand side
across all the equations, while $b$
is the total number of the quadratic monomials.
In this context, the $(7,2)$ Lorenz system
has been presented in $1963$~\cite{Lorenz},
containing $7$ monomials of which only $2$ are quadratic, 
and in $1976$ the $(7,1)$ R\"ossler system~\cite{Rossler},
containing only $1$ quadratic.
Subsequently, even simpler chaotic systems 
have been discovered: 
five $(5,2)$ and 
fourteen $(6,1)$ DSs in $1994$~\cite{Sprott},
and a $(5,1)$ DS in 1997~\cite{Sprott_51}.
To this date, no simpler 
three-dimensional quadratic DSs
with chaos have been reported;
instead, the literature has focused more on 
generally more complicated
chaotic DSs with various additional properties, 
such as those with no equilibria~\cite{Sprott_0_eq}, 
with a unique and stable equilibrium~\cite{Sprott_1_eq} 
and with a line of equilibria~\cite{Sprott_Line}. 

However, none of these simple chaotic DSs
are also \emph{chemical dynamical systems} (CDSs)~\cite{QCM}
- a special subset of polynomial DSs 
that can model the time-evolution
of the (non-negative) concentrations of chemical species
reacting according to a set of mass-action 
reactions jointly
called a \emph{chemical reaction network} 
(CRN)~\cite{Janos,Feinberg}. 
Put differently, while the reported chaotic 
DSs can be experimentally implemented 
with mechanical or electronic devices~\cite{Sprott_book_2}, 
they cannot be implemented with 
chemical reactions - a task achievable only for CDSs,
using e.g. DNA molecules~\cite{DNA}.
A fundamental property preventing this implementation, 
that distinguishes CDSs
from general polynomial DSs, 
is the absence of negative monomials
that do not include as a factor the corresponding 
variable; for example, $-1$
does not include a factor of $x$, making
$\mathrm{d} x/\mathrm{d} t = - 1$ a non-chemical DS. 
In this paper, we associate to every CRN, 
induced by a quadratic CDS, 
a vector $(c,d)$, where $c$ is the total
number of reactions, and $d$
the total number of the quadratic reactions 
- those with two molecules on the left-hand side,
such as $X + Y \xrightarrow[]{} Z$,
which correspond to the quadratic monomials.

Compared to general DSs, the number of 
reported three-dimensional quadratic CDSs 
with chaos is significantly lower. 
In this context, in 1980, reported are 
a $(12,9)$ CDS with a $(10,7)$ 
CRN~\cite{LV2}, and a $(9,6)$ CDS with a $(7,4)$ CRN~\cite{RosslerW}, 
known as the \emph{minimal Willamowski–R\"ossler} system.
Both of these CDSs take the \emph{Lotka-Volterra}
form~\cite{LV}: all the monomials in a given equation 
contain as a factor the corresponding variable.
More recently, a systematic method has been developed that,
under suitable robustness condition,
allows one to map 
every chaotic three-dimensional quadratic $(a,1)$ DS,
with any $a$, to a three-dimensional \emph{quadratic} 
CDS with chaos preserved~\cite{Chaos_1,QCM}.
Using this method, the following 
relatively simple systems
have been constructed in~\cite{Chaos_1}: 
two $(11,5)$ CDSs with $(9,4)$ CRNs, 
one resembling the  R\"ossler system, 
while the other has a chaotic attractor and a unique
and stable equilibrium, resembling system 
$\textrm{SE}_{17}$ from~\cite{Sprott_1_eq},
and a $(10,3)$ CDS with an $(8,3)$ CRN
resembling system P from~\cite{Sprott}.
To the best of the authors' knowledge, 
these are the simplest CDSs and CRNs with
chaos reported in the literature to this date.
For more general results, and 
references about higher-dimensional 
and higher-degree CDSs
with chaos, see~\cite{Chaos_1}.

A natural question arises: 
What are structurally the simplest CDSs with chaos?
In particular, is chaos possible in
three-dimensional quadratic CDSs 
with fewer than $9$ monomials, 
or fewer than $3$ quadratic monomials?
Are there chaotic CRNs that have fewer than 
$7$ reactions, or fewer than $3$ quadratic reactions?
In this paper, we address these questions. 
In particular, we first derive 
two elementary theoretical 
results. Included in the first result is that
if an $N$-dimensional polynomial DS has 
a compact invariant set in the positive orthant, 
then it has at least $2 N$ monomials, with at least
one negative and one positive monomial per equation, 
see Theorem~\ref{theorem:sign}. This result implies that 
three-dimensional CDSs with chaos in the positive orthant
have at least $6$ monomials.
The second result is that three-dimensional 
CDSs with chaos have a negative nonlinear monomial
with at least two distinct factors, 
see Theorem~\ref{theorem:non_linearity}.
We then perform an extensive computer search
for chaotic three-dimensional quadratic
CDSs with only $6$ or $7$ monomials,
using the derived theoretical results 
to reduce the computational intensity of the search.
We discover $20$  chemical systems,
named $\textrm{CS}_1$--$\textrm{CS}_{20}$
and presented in Tables~\ref{tab:one}--\ref{tab:two}, 
with numerically detected chaos displayed 
in Figure~\ref{fig:1}.
In particular, we show that chaos is 
detected in as simple as 
$(6,4)$ and $(7,2)$ CDSs, and 
$(4,3)$ and $(5,2)$ CRNs,
thus quantifying the 
structural complexity of chaotic CDSs. 
Furthermore, our results indicate that CDSs with chaos,
while structurally more complicated,
are nevertheless ubiquitous.

The paper is organized as follows.
In Section~\ref{sec:background}, 
we present some background theory;
further details can be found 
in Appendix~\ref{app:CRN}.
In Section~\ref{sec:theory}, we present
some theoretical results about chaotic CDSs,
which we then use in Section~\ref{sec:examples} 
to narrow down the computer search for
simple three-dimensional quadratic CDSs with chaos.
Finally, we close the paper with some open problems 
in Section~\ref{sec:discussion}.

\begin{table}[!htbp]
\small
\caption{\it{\emph{$20$ chemical systems with chaos.}} 
Each row shows one chemical system. The first column
contains the name of the system, and the structural 
complexity of its \emph{CDS} and \emph{CRN},
shown in the second and third columns,
respectively. 
The remaining columns respectively show
the non-negative equilibria, 
the Lyapunov exponents (\emph{LE}s), 
the Lyapunov dimension (\emph{LD}), 
and an initial condition
(\emph{IC}) attracted to the chaotic set.} 
\setlength{\tabcolsep}{6pt} 
\renewcommand{\arraystretch}{1.5} 
\begin{tabular}{c l l l r l l} 
\hline\hline
System & CDS & CRN & Equilibria & LEs & LD & IC \\ 
\hline & \\ [-2.0ex] 
$\textrm{CS}_1$ 
& $\frac{\mathrm{d} x}{\mathrm{d} t} = x^2 - 0.5 x y$
& $2 X \xrightarrow[]{1} 3 X + Y$,
& $\mathrm{nh}_{1,0}(0,0,0)$
& $0.1597$
& 
& $4$ \\
$(6,4)$
& $\frac{\mathrm{d} y}{\mathrm{d} t} = x^2 - y z$
& $X + Y \xrightarrow[]{0.5} Y$, 
$Y + Z \xrightarrow[]{1} Z$,
& 
& $0$
& $2.0228$
& $2$ \\
$(5,3)$ 
& $\frac{\mathrm{d} z}{\mathrm{d} t} = y - 0.9 z$
& $Y \xrightarrow[]{1} Y + Z$, $Z \xrightarrow[]{0.9} \varnothing$ 
&
& $-6.9947$
& 
& $3$ \\ [1ex]
\hline
$\textrm{CS}_2$ 
& $\frac{\mathrm{d} x}{\mathrm{d} t} = x y - 0.4 x z$
& $X + Y \xrightarrow[]{1} 2 X + Y + Z$,
& $\mathrm{nh}_{2,0}(0,0,0)$
& $0.2923$
& 
& $7$ \\
$(6,4)$ 
& $\frac{\mathrm{d} y}{\mathrm{d} t} = -y + x^2$
& $X + Z \xrightarrow[]{0.4} Z$, $Y \xrightarrow[]{1} \varnothing$,
& $\textrm{sf}_{1,2}(5,25,62.5)$
& $0$
& $2.0888$
& $25$ \\
$(5,3)$
& $\frac{\mathrm{d} z}{\mathrm{d} t} = -2 z + x y$
& $2 X \xrightarrow[]{1} 2 X + Y$, $Z \xrightarrow[]{2} \varnothing$ 
&
& $-3.2923$
& 
& $57$ \\ [1ex]
\hline
$\textrm{CS}_3$ 
& $\frac{\mathrm{d} x}{\mathrm{d} t} = y^2 - x z$
& $2 Y \xrightarrow[]{1} X + 2 Y$,
& $\mathrm{nh}_{2,0}(0,0,0)$
& $0.0405$
& 
& $3$ \\
$(6,4)$ 
& $\frac{\mathrm{d} y}{\mathrm{d} t} = -y + 0.04 x^2$
& $X + Z \xrightarrow[]{1} 2 Z$, $Y \xrightarrow[]{1} \varnothing$,
& $\textrm{sf}_{1,2}(5.4,1.17,0.25)$
& $0$
& $2.0260$
& $2$ \\
$(5,3)$ 
& $\frac{\mathrm{d} z}{\mathrm{d} t} = - 5.4 z + x z$
& $2 X \xrightarrow[]{0.04} 2 X + Y$, 
$Z \xrightarrow[]{5.4} \varnothing$
&
& $-1.5549$
& 
& $1$ \\ [1ex]
\hline
$\textrm{CS}_4$ 
& $\frac{\mathrm{d} x}{\mathrm{d} t} = -0.01 x^2 + y z$
& $2 X \xrightarrow[]{0.01} X$, 
$Y + Z \xrightarrow[]{1} X + Z$,
& $\textrm{nh}_{1,0}(0,0,0)$
& $0.3506$
& 
& $200$ \\
$(6,5)$ 
& $\frac{\mathrm{d} y}{\mathrm{d} t} = 0.27 x y - y z$
& $X + Y \xrightarrow[]{0.27} X + 2 Y$,
& $\mathrm{sf}_{1,2}(196.83,7.29,53.14)$
& $0$
& $2.0738$
& $6$ \\
$(5,4)$ 
& $\frac{\mathrm{d} z}{\mathrm{d} t} = -z + y^2$
& $Z \xrightarrow[]{1} \varnothing$, 
$2 Y \xrightarrow[]{1} 2 Y + Z$
&
& $-4.7500$
& 
& $55$ \\ [1ex]
\hline
$\textrm{CS}_5$ 
& $\frac{\mathrm{d} x}{\mathrm{d} t} = -0.1 x^2 + x y$
& $2 X \xrightarrow[]{0.1} X$, 
$X + Y \xrightarrow[]{1} 2 X + 2 Y$,
& $\textrm{nh}_{1,0}(0,0,0)$
& $0.1538$
& 
& $17$ \\
$(6,5)$ 
& $\frac{\mathrm{d} y}{\mathrm{d} t} = x y - y z$
& $Y + Z \xrightarrow[]{1} Z$,
$Z \xrightarrow[]{1} \varnothing$,
& $\mathrm{sf}_{1,2}(19.61,1.96,19.61)$
& $0$
& $2.0599$
& $3$ \\
$(5,4)$ 
& $\frac{\mathrm{d} z}{\mathrm{d} t} = -z + 5.1 y^2$
& $2 Y \xrightarrow[]{5.1} 2 Y + Z$
&
& $-2.5681$
& 
& $16$ \\ [1ex]
\hline
$\textrm{CS}_6$ 
& $\frac{\mathrm{d} x}{\mathrm{d} t} = 2 y^2 - x y$
& $2 Y \xrightarrow[]{2} X + 2 Y$,
& $\textrm{nh}_{1,0}(0,0,0)$
& $0.0424$
& 
& $10$ \\
$(6,5)$ 
& $\frac{\mathrm{d} y}{\mathrm{d} t} = x y - 0.5 y z$
& $X + Y \xrightarrow[]{1} 2 Y + Z$,
& $\textrm{nh}_{1,1}(x > 0,0,0)$ 
& $0$
& $2.0157$
& $1$ \\
$(4,3)$
& $\frac{\mathrm{d} z}{\mathrm{d} t} = - z + x y$
& $Y + Z \xrightarrow[]{0.5} Z$, 
$Z \xrightarrow[]{1} \varnothing$
& $\mathrm{sf}_{1,2}(4,2,8)$
& $-2.7004$
& 
& $10$ \\ [1ex]
\hline
$\textrm{CS}_7$ 
& $\frac{\mathrm{d} x}{\mathrm{d} t} = -0.1 x z + y z$
& $X + Z \xrightarrow[]{0.1} Z$, 
$Y + Z \xrightarrow[]{1} X + Z$,
& $\textrm{nh}_{1,0}(0,0,0)$ 
& $0.1916$
& 
& $15$ \\
$(6,5)$ 
& $\frac{\mathrm{d} y}{\mathrm{d} t} = x y - y z$
& $X + Y \xrightarrow[]{1} X + 2 Y$,
& $\textrm{nh}_{1,1}(x > 0,0,0)$
& $0$
& $2.0753$
& $2$ \\
$(5,4)$ 
& $\frac{\mathrm{d} z}{\mathrm{d} t} = - z + 5 y^2$
& $Z \xrightarrow[]{1} \varnothing$, 
$2 Y \xrightarrow[]{5} 2 Y + Z$
& $\mathrm{sf}_{1,2}(20,2,20)$
& $-2.5448$
& 
& $14$ \\ [1ex]
\hline
$\textrm{CS}_8$ 
& $\frac{\mathrm{d} x}{\mathrm{d} t} = x y - 0.6 x z$
& $X + Y \xrightarrow[]{1} 2 X + 2 Y$,
& $\textrm{nh}_{1,0}(0,0,0)$ 
& $0.0662$
& 
& $5$ \\
$(6,5)$ 
& $\frac{\mathrm{d} y}{\mathrm{d} t} = x y - 4 y z$
& $X + Z \xrightarrow[]{0.6} Z$,
$Y + Z \xrightarrow[]{4} Z$,
& $\textrm{nh}_{1,1}(x > 0,0,0)$
& $0$
& $2.0621$
& $2$ \\
$(5,4)$ 
& $\frac{\mathrm{d} z}{\mathrm{d} t} = -z + y^2$
& $Z \xrightarrow[]{1} \varnothing$, 
$2 Y \xrightarrow[]{1} 2 Y + Z$
& $\mathrm{sf}_{1,2}(11.11,1.67,2.78)$
& $-1.0662$
& 
& $1$ \\ [1ex]
\hline\hline
\end{tabular}
\label{tab:one} 
\end{table}

\begin{table}[!htbp]
\vskip -2.7cm
\small
\caption{Table~\ref{tab:one} continued.} 
\setlength{\tabcolsep}{6.0pt} 
\renewcommand{\arraystretch}{1.5} 
\begin{tabular}{c l l l r l l} 
\hline\hline
System & CDS & CRN & Equilibria & LEs & LD & IC \\ 
\hline & \\ [-2.0ex] 
$\textrm{CS}_9$ 
& $\frac{\mathrm{d} x}{\mathrm{d} t} = -4 x + x y$
& $X \xrightarrow[]{4} Z$,
$Y + Z \xrightarrow[]{0.2} Z$,
& $\textrm{s}_{2,1}(0,0,0)$ 
& $0.1659$
& 
& $2$ \\
$(7,2)$ 
& $\frac{\mathrm{d} y}{\mathrm{d} t} = 6 y - 0.2 y z$
& $X + Y \xrightarrow[]{1} 2 X + Y$,
& $\mathrm{sf}_{2,1}(0,10,30)$
& $0$
& $2.0766$
& $1$ \\
$(5,2)$ 
& $\frac{\mathrm{d} z}{\mathrm{d} t} = 4 x + 6 y - 2 z$
& $Y \xrightarrow[]{6} 2 Y + Z$, 
$Z \xrightarrow[]{2} \varnothing$
& $\mathrm{sf}_{1,2}(9,4,30)$
& $-2.1659$
& 
& $22$ \\ [1ex]
\hline
$\textrm{CS}_{10}$ 
& $\frac{\mathrm{d} x}{\mathrm{d} t} = -x + z + y^2$
& $X \xrightarrow[]{1} \varnothing$,
$Z \xrightarrow[]{1} X$, 
& $\textrm{s}_{2,1}(0,0,0)$ 
& $0.0930$
& 
& $3$ \\
$(7,3)$ 
& $\frac{\mathrm{d} y}{\mathrm{d} t} = 3 y - x y$
& $2 Y \xrightarrow[]{1} X + 2 Y$,
$Y \xrightarrow[]{3} 2 Y$,
& $\mathrm{sf}_{2,1}(3,\sqrt{3},0)$ 
& $0$
& $2.0851$
& $3$ \\
$(6,3)$ 
& $\frac{\mathrm{d} z}{\mathrm{d} t} = -z + 2 y z$
& $X + Y \xrightarrow[]{1} X$, 
$Y + Z \xrightarrow[]{2} Y + 2 Z$
& $\mathrm{sf}_{1,2}(3,0.5,2.75)$
& $-1.0930$
& 
& $1$ \\ [1ex]
\hline
$\textrm{CS}_{11}$ 
& $\frac{\mathrm{d} x}{\mathrm{d} t} = -x + 8 z + y^2$
& $X \xrightarrow[]{1} \varnothing$,
$Z \xrightarrow[]{8} X + 2 Z$, 
& $\textrm{s}_{2,1}(0,0,0)$ 
& $0.0930$
& 
& $10$ \\
$(7,3)$ 
& $\frac{\mathrm{d} y}{\mathrm{d} t} = - y + y z$
& $2 Y \xrightarrow[]{1} X + 2 Y$,
$X + Z \xrightarrow[]{0.3} X$,
& $\mathrm{sf}_{2,1}(26.67,0,3.33)$ 
& $0$
& $2.0851$
& $1$ \\
$(6,3)$ 
& $\frac{\mathrm{d} z}{\mathrm{d} t} = 8 z - 0.3 x z$
& $Y \xrightarrow[]{1} \varnothing$, 
$Y + Z \xrightarrow[]{1} 2 Y + Z$
& $\mathrm{sf}_{1,2}(26.67,4.32,1)$
& $-1.0930$
& 
& $1$ \\ [1ex]
\hline
$\textrm{CS}_{12}$ 
& $\frac{\mathrm{d} x}{\mathrm{d} t} = 3 z + x^2 - x y$
& $Z \xrightarrow[]{3} X$,
$2 X\xrightarrow[]{1} 3 X$,
& $\textrm{nh}_{2,0}(0,0,0)$ 
& $0.1185$
& 
& $5$ \\
$(7,3)$ 
& $\frac{\mathrm{d} y}{\mathrm{d} t} = -4 y + x y$
& $X + Y \xrightarrow[]{1} 2 Y$, 
$Y \xrightarrow[]{4} \varnothing$,
& $\mathrm{sf}_{1,2}(4,5.33,1.78)$ 
& $0$
& $2.0234$
& $5$ \\
$(5,2)$ 
& $\frac{\mathrm{d} z}{\mathrm{d} t} = y - 3 z$
& $Y \xrightarrow[]{1} Y + Z$
&
& $-5.0597$
& 
& $2$ \\ [1ex]
\hline
$\textrm{CS}_{13}$ 
& $\frac{\mathrm{d} x}{\mathrm{d} t} = y + x^2 - x z$
& $Y \xrightarrow[]{1} X + Y$,
$2 X \xrightarrow[]{1} 3 X$,
& $\textrm{nh}_{2,0}(0,0,0)$ 
& $0.6042$
& 
& $5$ \\
$(7,4)$ 
& $\frac{\mathrm{d} y}{\mathrm{d} t} = -4 y + x z$
& $X + Z \xrightarrow[]{1} Y + 2 Z$,
& $\mathrm{sf}_{1,2}(9,27,12)$ 
& $0$
& $2.0352$
& $38$ \\
$(5,2)$ 
& $\frac{\mathrm{d} z}{\mathrm{d} t} =  -9 z + x z$
& $Y \xrightarrow[]{4} \varnothing$,
$Z \xrightarrow[]{9} \varnothing$
& 
& $-17.1720$
& 
& $9$ \\ [1ex]
\hline
$\textrm{CS}_{14}$ 
& $\frac{\mathrm{d} x}{\mathrm{d} t} = 0.2 y + x^2 - x z$
& $Y \xrightarrow[]{0.2} X + Y$,
$2 X \xrightarrow[]{1} 3 X$,
& $\textrm{nh}_{1,0}(0,0,0)$ 
& $0.0683$
& 
& $1$ \\
$(7,4)$ 
& $\frac{\mathrm{d} y}{\mathrm{d} t} = z - x y$
& $X + Z \xrightarrow[]{1} 2 Z$, 
$Z \xrightarrow[]{1} Y$
& $\mathrm{sf}_{1,2}(1,1.25,1.25)$
& $0$
& $2.0790$
& $2$ \\
$(5,3)$ 
& $\frac{\mathrm{d} z}{\mathrm{d} t} = -z + x z$
& $X + Y \xrightarrow[]{1} X$
& 
& $-0.8639$
& 
& $1$ \\ [1ex]
\hline
$\textrm{CS}_{15}$ 
& $\frac{\mathrm{d} x}{\mathrm{d} t} = x + y^2 - x z$
& $X \xrightarrow[]{1} 2 X$,
$2 Y \xrightarrow[]{1} X + 2 Y$,
& $\textrm{n}_{0,3}(0,0,0)$ 
& $0.0212$
& 
& $0$ \\
$(7,4)$ 
& $\frac{\mathrm{d} y}{\mathrm{d} t} = 2 y - x y$
& $X + Z\xrightarrow[]{1} Z$, 
$Y \xrightarrow[]{2} 2 Y$,
$Z \xrightarrow[]{0.2} 2 Z$,
& $\mathrm{sf}_{1,2}(2,0.4,1.08)$
& $0$
& $2.0158$
& $1$ \\
$(7,4)$ 
& $\frac{\mathrm{d} z}{\mathrm{d} t} = 0.2 z - 0.5 y z$
& $X + Y \xrightarrow[]{1} X$, 
$Y + Z \xrightarrow[]{0.5} Y$
& 
& $-1.3397$
& 
& $4$ \\ [1ex]
\hline
$\textrm{CS}_{16}$ 
& $\frac{\mathrm{d} x}{\mathrm{d} t} = 0.28 - x y + y z$
& $\varnothing \xrightarrow[]{0.28} X$,
$X + Y \xrightarrow[]{1} 2 Y$,
& $\textrm{sf}_{2,1}(1,0.28,0)$ 
& $0.0293$
& 
& $3$ \\
$(7,4)$ 
& $\frac{\mathrm{d} y}{\mathrm{d} t} = -y + x y$
& $Y + Z \xrightarrow[]{1} X + Y$,
& $\mathrm{sf}_{1,2}(1,1,0.72)$ 
& $0$
& $2.0285$
& $1$ \\
$(5,2)$ 
& $\frac{\mathrm{d} z}{\mathrm{d} t} = z - y z$
& $Y \xrightarrow[]{1} \varnothing$, 
$Z \xrightarrow[]{1} 2 Z$
& 
& $-1.0293$
& 
& $5$ \\ [1ex]
\hline
$\textrm{CS}_{17}$ 
& $\frac{\mathrm{d} x}{\mathrm{d} t} = 3 x^2 + 0.04 z^2 - x y$
& $2 X \xrightarrow[]{3} 3 X + Y$,
& $\textrm{nh}_{2,0}(0,0,0)$ 
& $0.1687$
& 
& $1$ \\
$(7,4)$ 
& $\frac{\mathrm{d} y}{\mathrm{d} t} = -y + 3 x^2$
& $2 Z \xrightarrow[]{0.04} X + 2 Z$,
$Y \xrightarrow[]{1} Z$,
& $\mathrm{sf}_{1,2}(1.16,4.05,4.05)$ 
& $0$
& $2.0721$
& $3$ \\
$(5,3)$ 
& $\frac{\mathrm{d} z}{\mathrm{d} t} = y - z$
& $X + Y \xrightarrow[]{1} Y$, 
$Z \xrightarrow[]{1} \varnothing$
& $\mathrm{s}_{2,1}(7.17,154.28,154.28)$
& $-2.3410$
& 
& $4$ \\ [1ex]
\hline
$\textrm{CS}_{18}$ 
& $\frac{\mathrm{d} x}{\mathrm{d} t} = z + x^2 - 0.35 x y$
& $Z \xrightarrow[]{1} X$,
$X + Y \xrightarrow[]{0.35} Y$,
& $\textrm{nh}_{2,0}(0,0,0)$ 
& $0.0524$
& 
& $1$ \\
$(7,4)$ 
& $\frac{\mathrm{d} y}{\mathrm{d} t} = -0.4 y + x^2$
& $Y \xrightarrow[]{0.4} \varnothing$, 
$2 X \xrightarrow[]{1} 3 X + Y$
& $\mathrm{sf}_{1,2}(1.44,5.17,0.54)$
& $0$
& $2.0803$
& $5$ \\
$(5,3)$ 
& $\frac{\mathrm{d} z}{\mathrm{d} t} = -z + 0.02 y^2$
& $2 Y \xrightarrow[]{0.02} 2 Y + Z$
& $\mathrm{s}_{2,1}(5.56,77.33,119.59)$
& $-0.6533$
& 
& $1$ \\ [1ex]
\hline
$\textrm{CS}_{19}$ 
& $\frac{\mathrm{d} x}{\mathrm{d} t} = z + 0.3 x^2 - x y$
& $Z \xrightarrow[]{1} X$,
$2 X \xrightarrow[]{0.3} 3 X$,
& $\textrm{nh}_{2,0}(0,0,0)$ 
& $0.1082$
& 
& $2$ \\
$(7,4)$ 
& $\frac{\mathrm{d} y}{\mathrm{d} t} = -3 y + x y$
& $Y \xrightarrow[]{3} \varnothing$, 
$X + Y \xrightarrow[]{1} 2 Y$, 
& $\mathrm{sf}_{1,2}(3,1,0.3)$
& $0$
& $2.1319$
& $4$ \\
$(5,3)$ 
& $\frac{\mathrm{d} z}{\mathrm{d} t} = -z + 0.3 y^2$
& $2 Y \xrightarrow[]{0.3} 2 Y + Z$
& $\mathrm{sf}_{2,1}(3,9,24.3)$
& $-0.8204$
& 
& $2$ \\ [1ex]
\hline
$\textrm{CS}_{20}$ 
& $\frac{\mathrm{d} x}{\mathrm{d} t} = -x + 0.4 y^2 + z^2$
& $X \xrightarrow[]{1} \varnothing$,
$2 Y \xrightarrow[]{0.4} X + 2 Y$,
& $\textrm{s}_{2,1}(0,0,0)$ 
& $0.1147$
& 
& $6$ \\
$(7,4)$ 
& $\frac{\mathrm{d} y}{\mathrm{d} t} = 4 y - x y$
& $2 Z \xrightarrow[]{1} X + 2 Z$, 
$Y \xrightarrow[]{4} 2 Y$,
$Z \xrightarrow[]{1} \varnothing$,
& $\mathrm{sf}_{2,1}(4,3.16,0)$
& $0$
& $2.1029$
& $1$ \\
$(7,4)$ 
& $\frac{\mathrm{d} z}{\mathrm{d} t} = -z + y z$
& $X + Y \xrightarrow[]{1} X$,
$Y + Z \xrightarrow[]{1} Y + 2 Z$
& $\mathrm{sf}_{1,2}(4,1,1.90)$
& $-1.1147$
& 
& $3$ \\ [1ex]
\hline\hline
\end{tabular}
\label{tab:two} 
\end{table}

\begin{figure}[!htbp]
\vskip 0.0cm
\leftline{
\hskip 1.7cm $\textrm{CS}_1$ 
\hskip 3.5cm $\textrm{CS}_2$ 
\hskip 3.5cm $\textrm{CS}_3$
\hskip 3.5cm $\textrm{CS}_4$}
\centerline{
\hskip 0.0cm
\includegraphics[width=0.25\columnwidth]{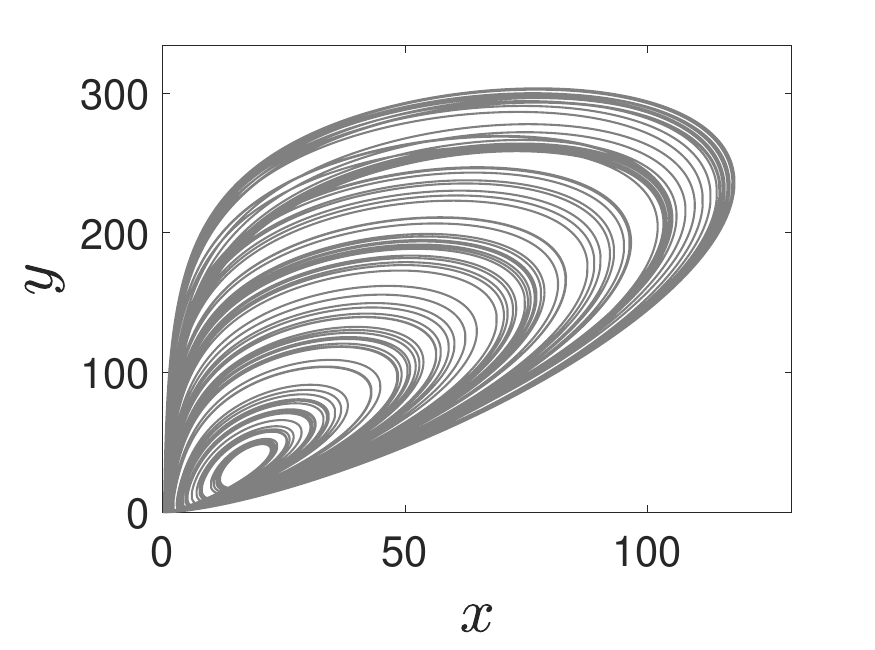}
\includegraphics[width=0.25\columnwidth]{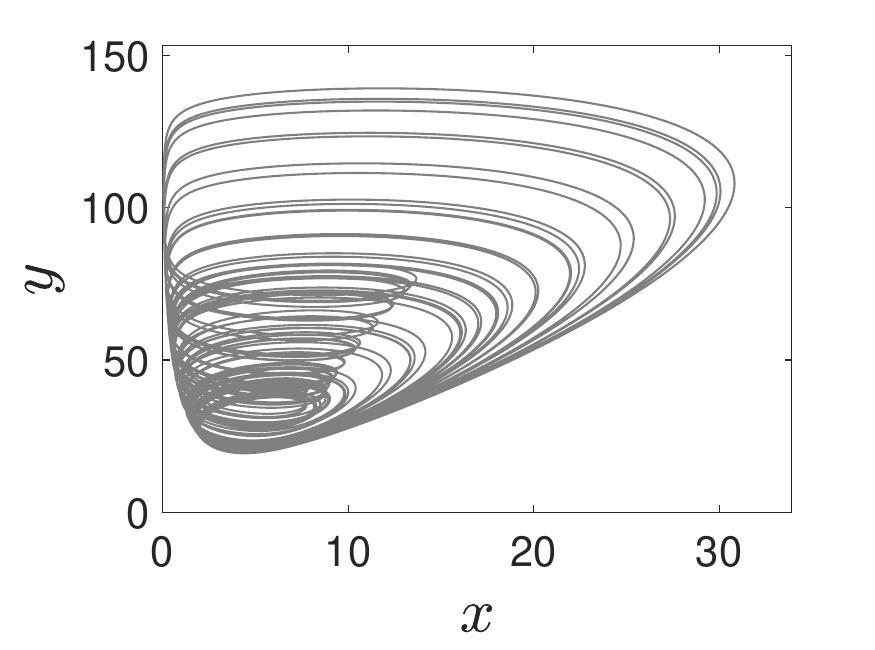}
\includegraphics[width=0.25\columnwidth]{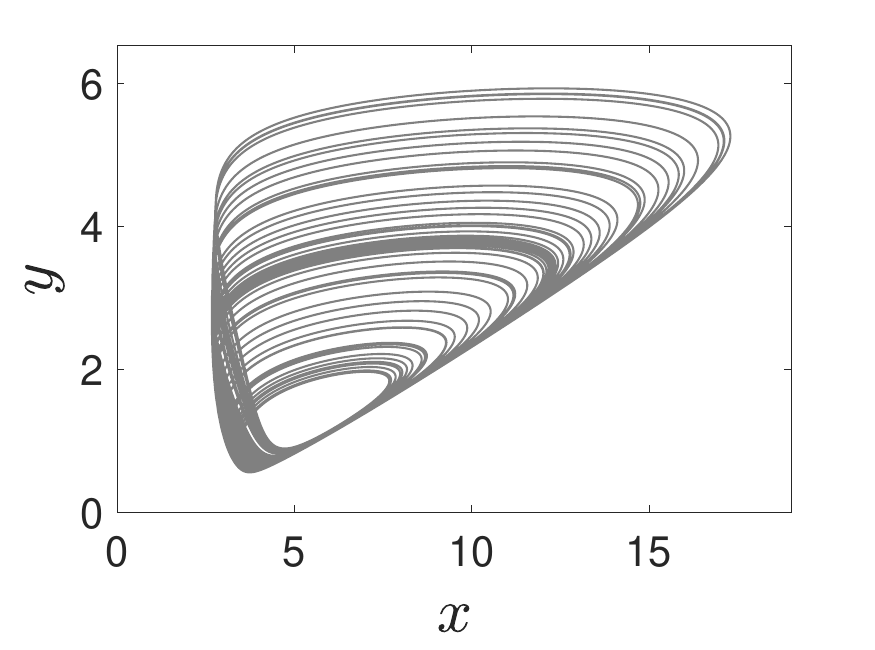}
\includegraphics[width=0.25\columnwidth]{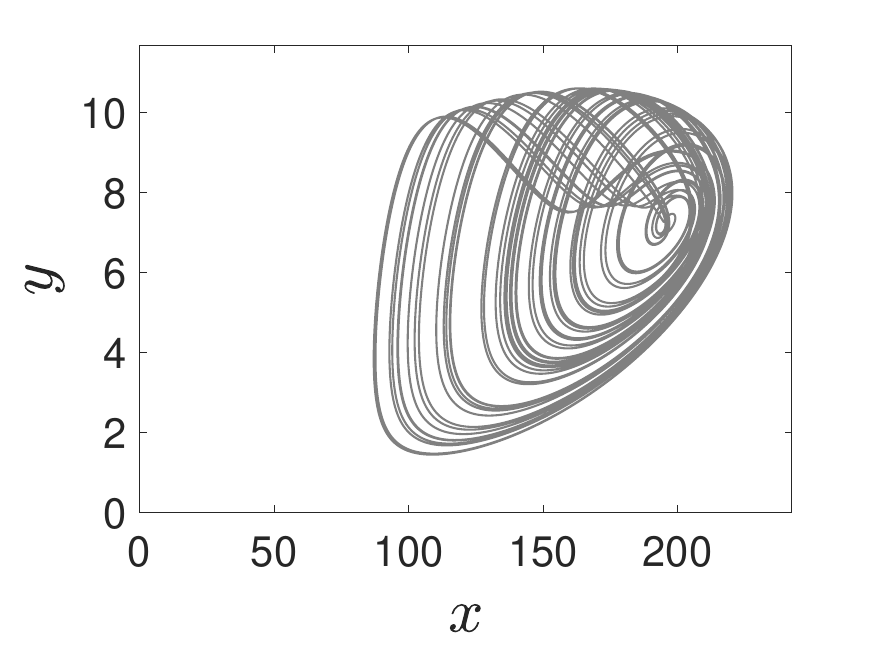}
}
\leftline{
\hskip 1.7cm $\textrm{CS}_5$ 
\hskip 3.5cm $\textrm{CS}_6$ 
\hskip 3.5cm $\textrm{CS}_7$
\hskip 3.5cm $\textrm{CS}_8$}
\centerline{
\hskip 0.0cm
\includegraphics[width=0.25\columnwidth]{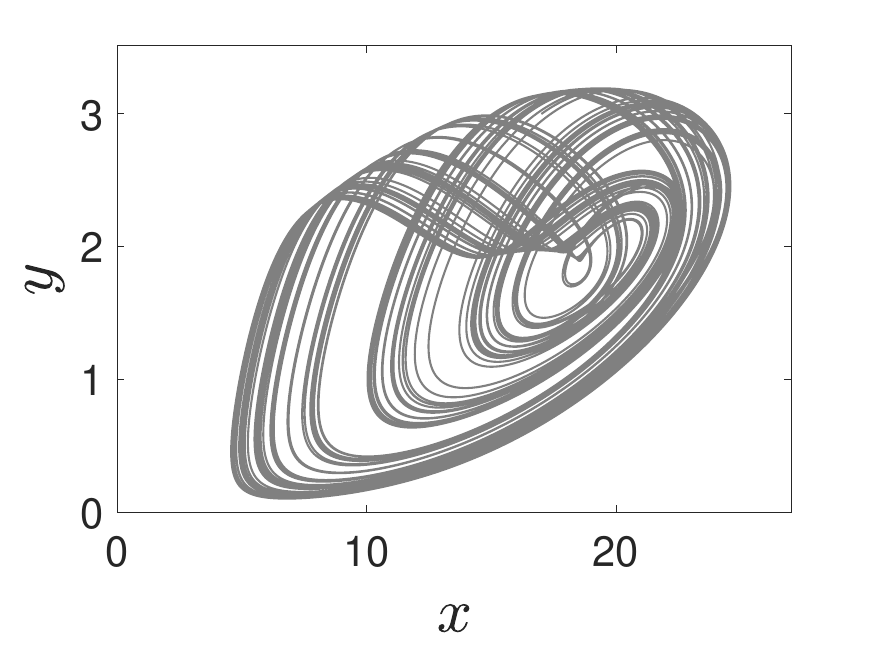}
\includegraphics[width=0.25\columnwidth]{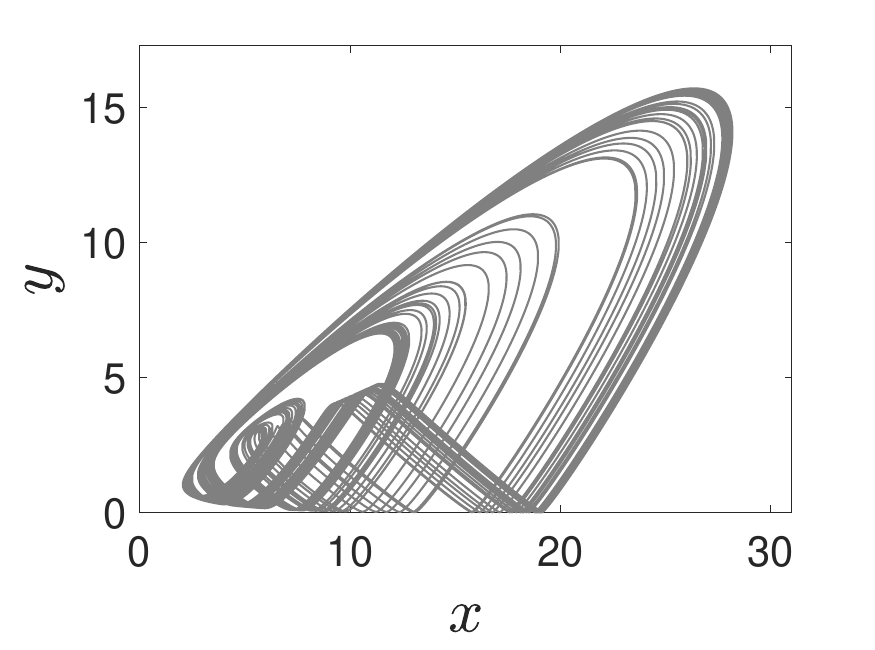}
\includegraphics[width=0.25\columnwidth]{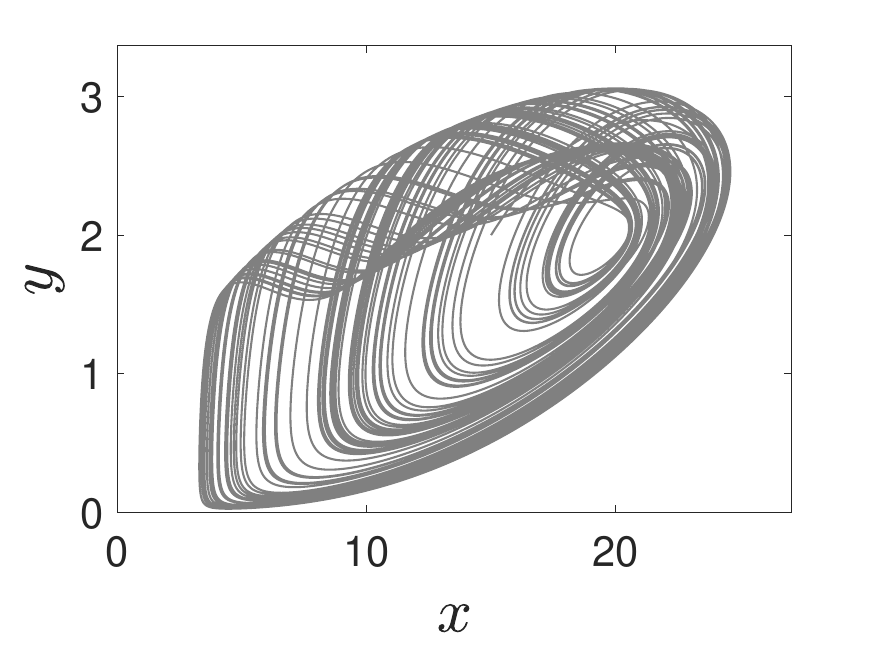}
\includegraphics[width=0.25\columnwidth]{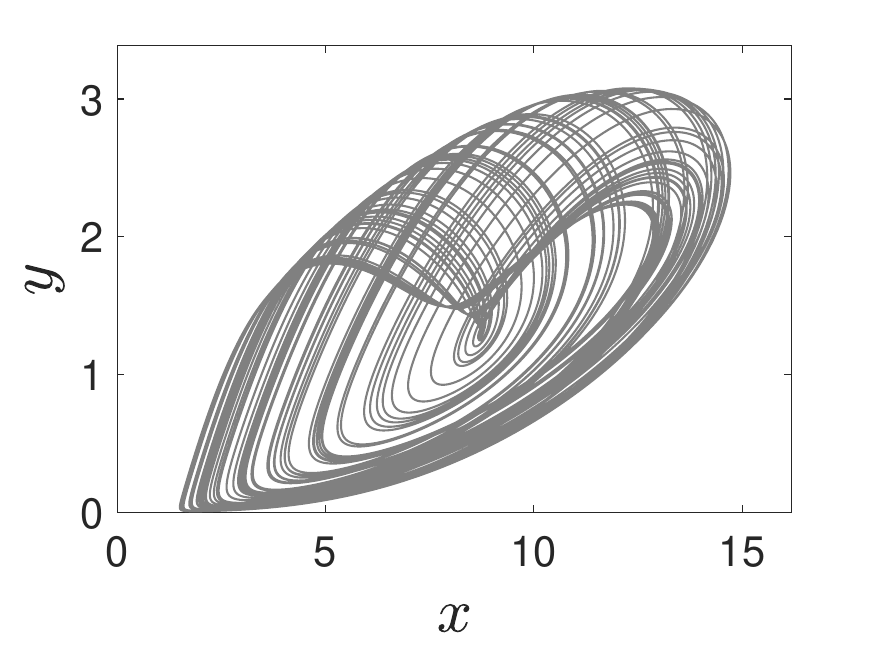}
}
\leftline{
\hskip 1.7cm $\textrm{CS}_9$ 
\hskip 3.5cm $\textrm{CS}_{10}$ 
\hskip 3.3cm $\textrm{CS}_{11}$
\hskip 3.3cm $\textrm{CS}_{12}$}
\centerline{
\hskip 0.0cm
\includegraphics[width=0.25\columnwidth]{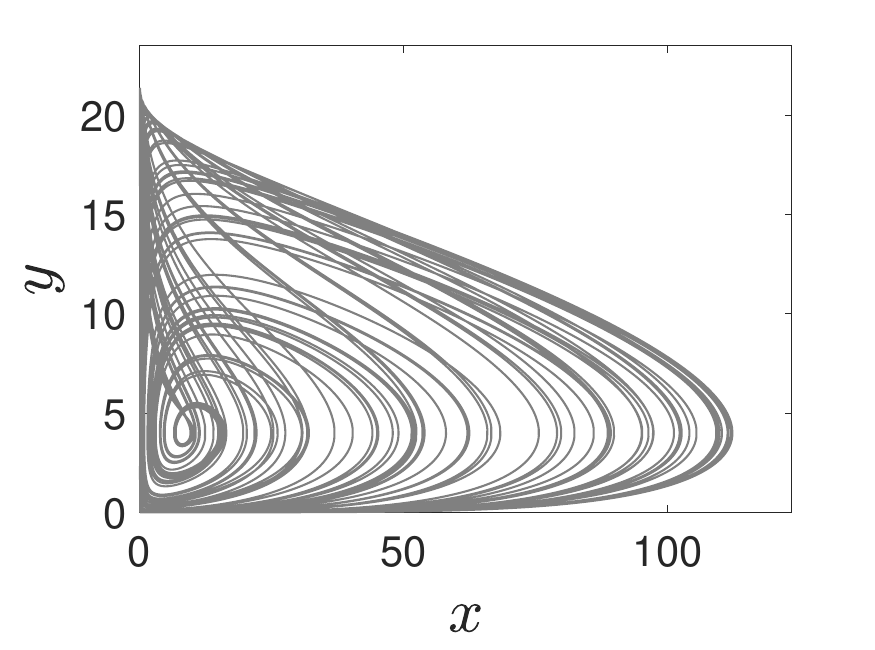}
\includegraphics[width=0.25\columnwidth]{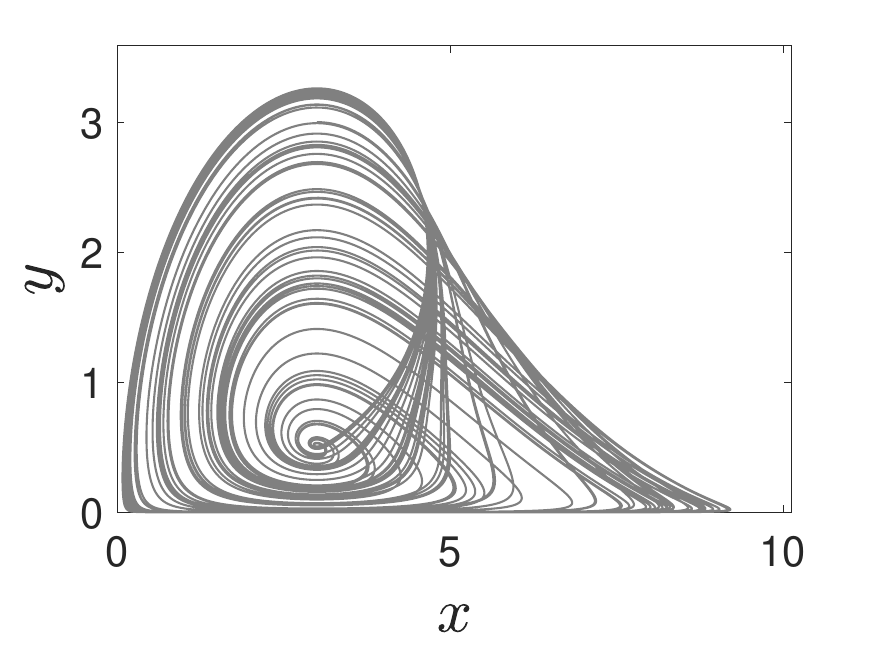}
\includegraphics[width=0.25\columnwidth]{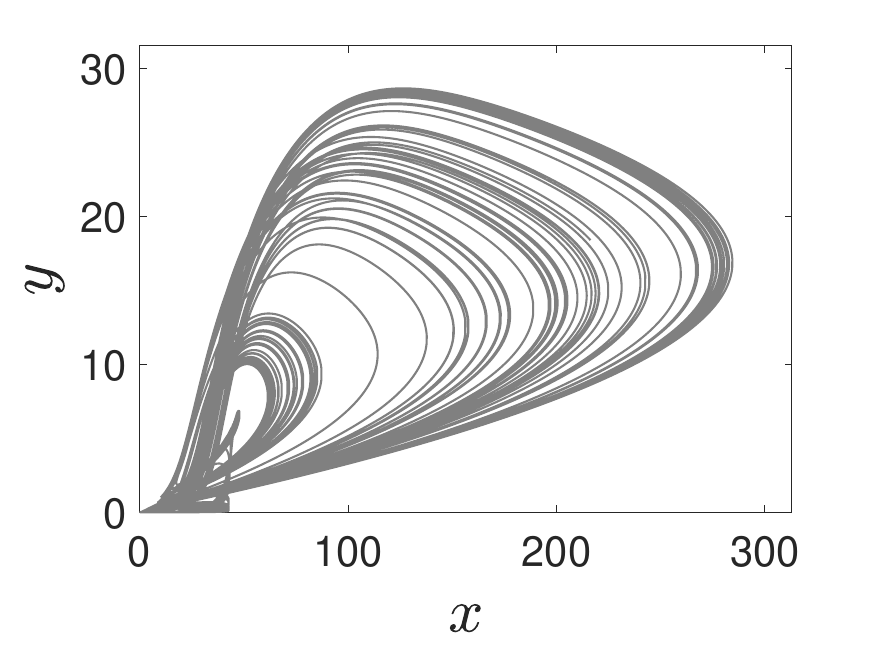}
\includegraphics[width=0.25\columnwidth]{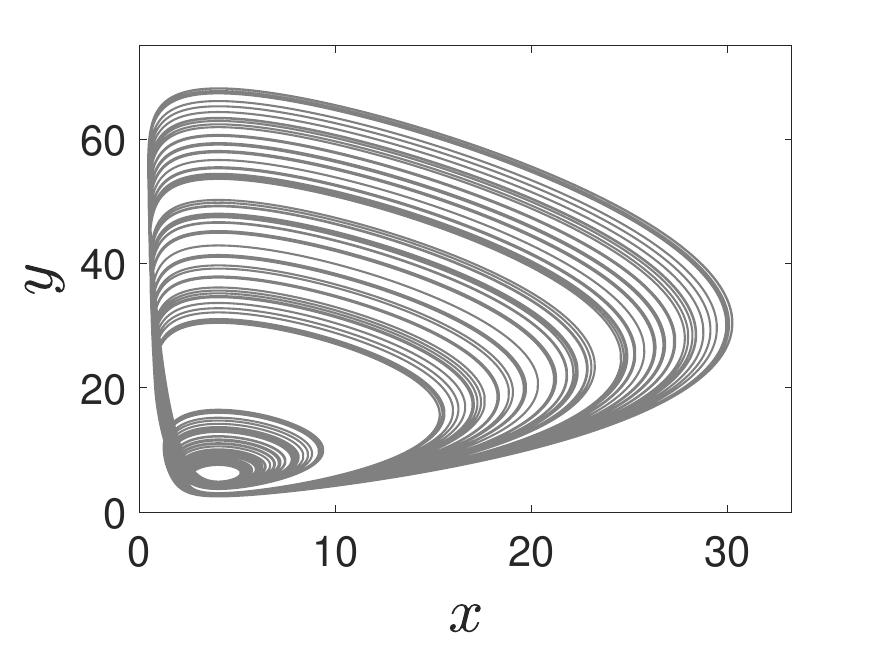}
}
\leftline{
\hskip 1.7cm $\textrm{CS}_{13}$ 
\hskip 3.3cm $\textrm{CS}_{14}$ 
\hskip 3.3cm $\textrm{CS}_{15}$
\hskip 3.3cm $\textrm{CS}_{16}$}
\centerline{
\hskip 0.0cm
\includegraphics[width=0.25\columnwidth]{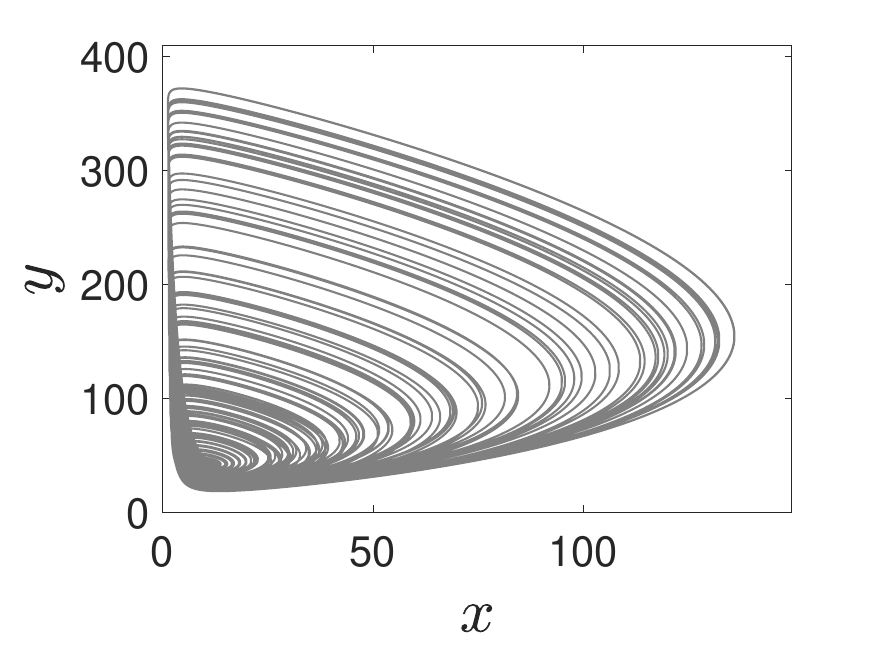}
\includegraphics[width=0.25\columnwidth]{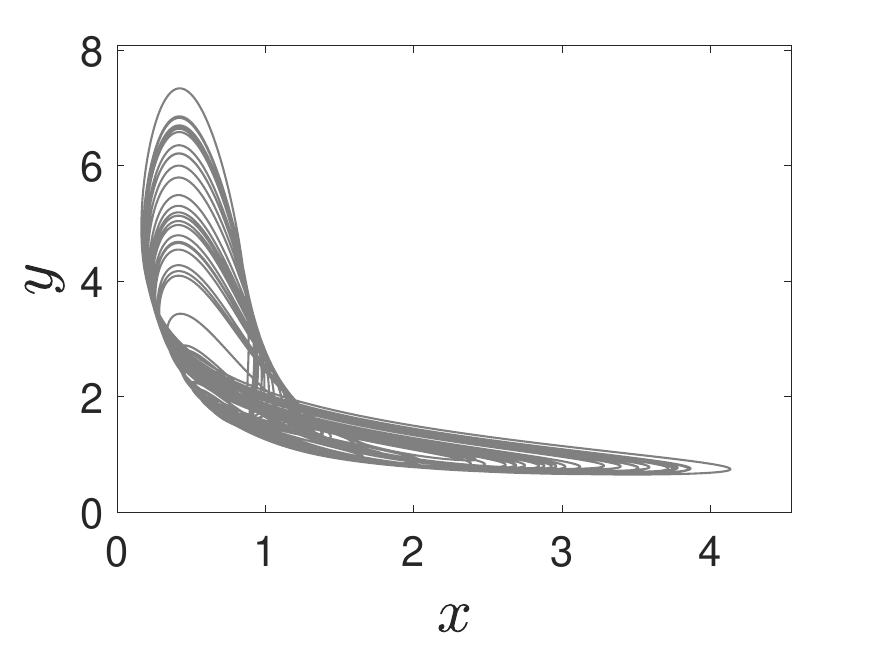}
\includegraphics[width=0.25\columnwidth]{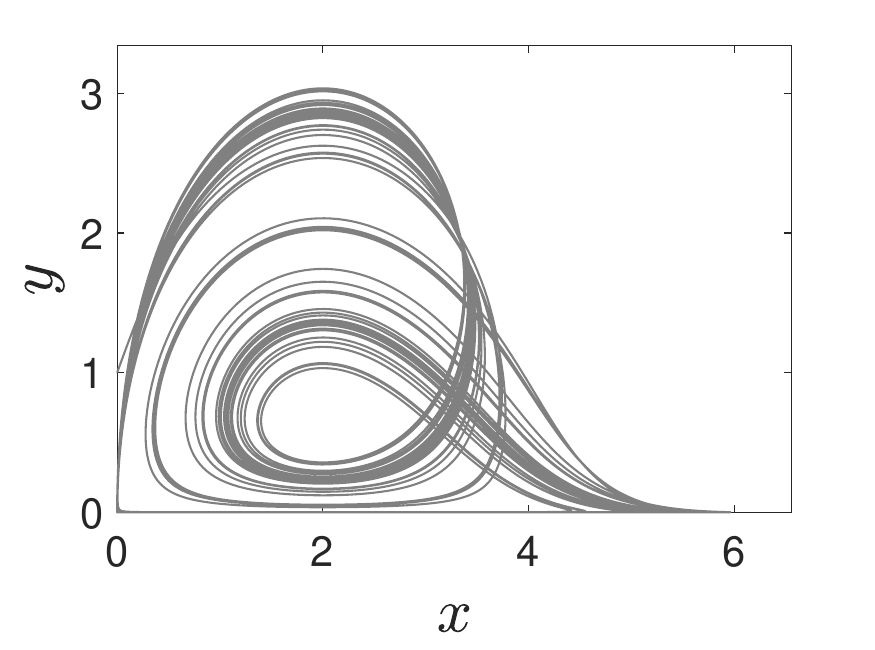}
\includegraphics[width=0.25\columnwidth]{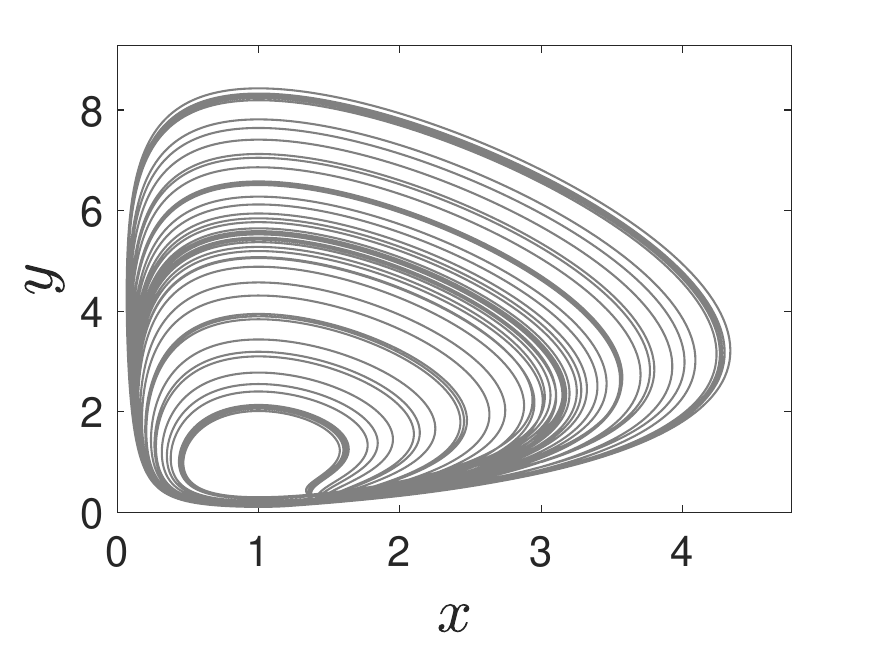}
}
\leftline{
\hskip 1.7cm $\textrm{CS}_{17}$ 
\hskip 3.3cm $\textrm{CS}_{18}$ 
\hskip 3.3cm $\textrm{CS}_{19}$
\hskip 3.3cm $\textrm{CS}_{20}$}
\centerline{
\hskip 0.0cm
\includegraphics[width=0.25\columnwidth]{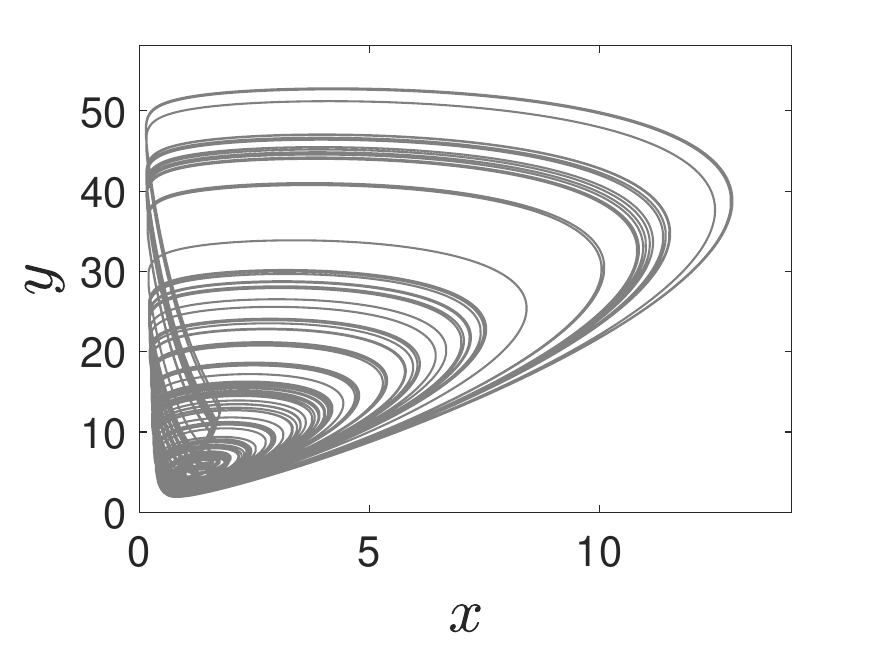}
\includegraphics[width=0.25\columnwidth]{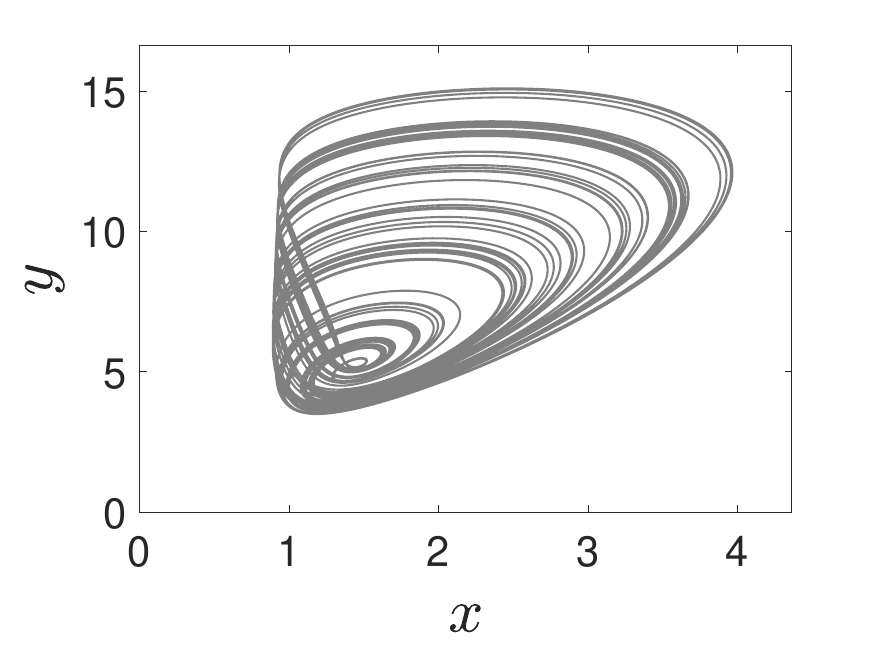}
\includegraphics[width=0.25\columnwidth]{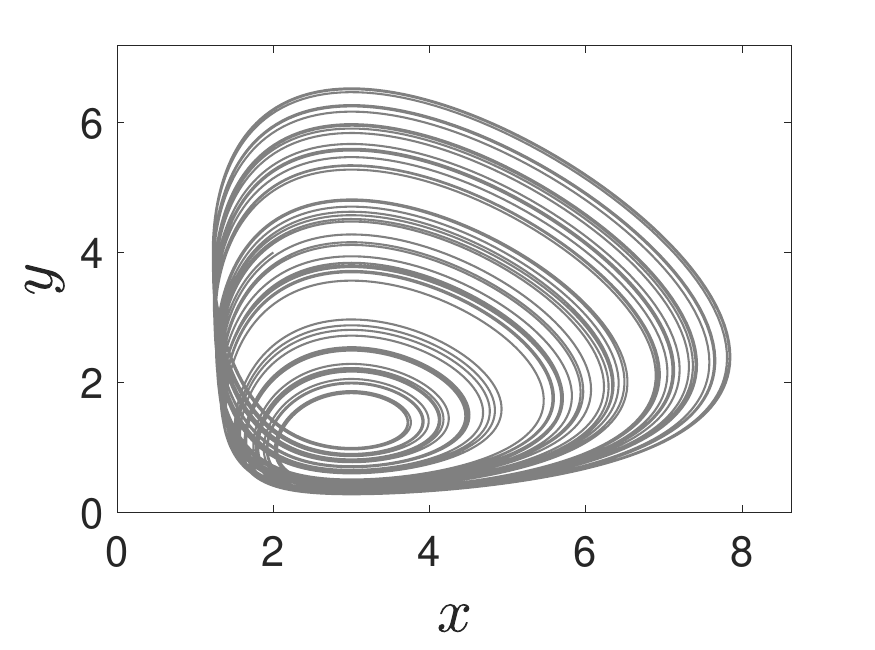}
\includegraphics[width=0.25\columnwidth]{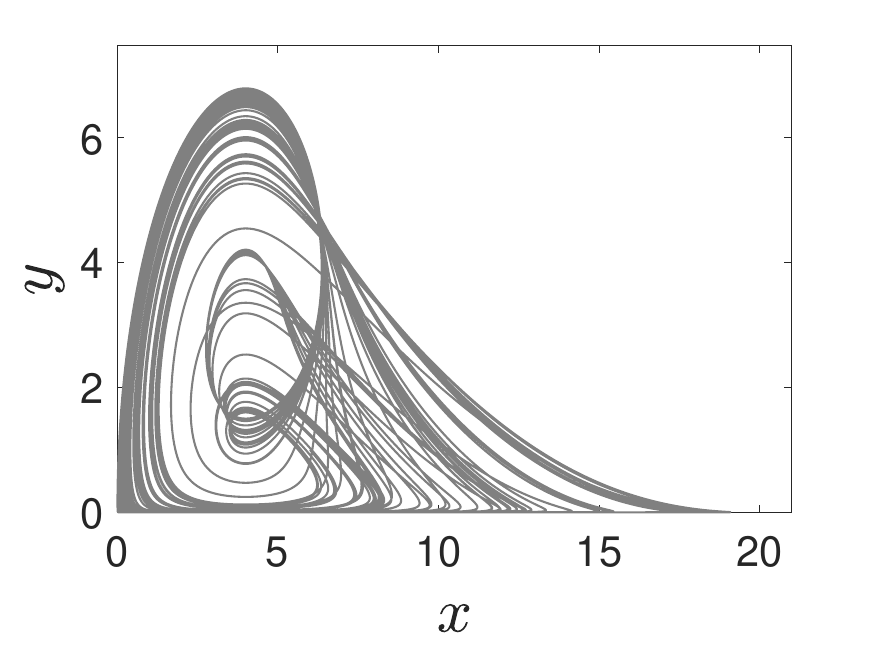}
}
\vskip -0.2cm
\caption{\it{\emph{Chemical systems with chaos.}} Numerically 
detected chaotic attractors in the $(x,y)$-space
for the $20$ chemical systems presented in 
\emph{Tables}~$\ref{tab:one}$--$\ref{tab:two}$.} 
\label{fig:1}
\end{figure}

\section{Background}
\label{sec:background}
In this section, we introduce some background theory;
for more details, see~\cite{Chaos_1}. 

\textbf{Notation}. The sets of all real, non-negative and 
positive real numbers are respectively denoted by $\mathbb{R}$, 
$\mathbb{R}_{\ge}$ and $\mathbb{R}_{>}$.
Column vectors are denoted by 
$\mathbf{x} = (x_1, x_2, \ldots, x_N)^{\top} 
\in \mathbb{R}^{N}$, where $\cdot^{\top}$ 
is the transpose operator.

\textbf{Dynamical systems}. 
Consider a system of autonomous
ordinary-differential equations
\begin{align}
\frac{\mathrm{d} x_i}{\mathrm{d} t} & = f_i(\mathbf{x}; n),
\; \; \; i = 1,2,\ldots, N,
\label{eq:DS}
\end{align}
where $t \in \mathbb{R}$ is time,
$\mathbf{x} = (x_1,x_2,\ldots,x_N)^{\top} \in \mathbb{R}^N$
and $f_i(\cdot; n) : \mathbb{R}^N \to \mathbb{R}$
a polynomial of degree at most $n$ that is not identically zero.
We say that~(\ref{eq:DS}) is an $N$-dimensional 
$n$-degree polynomial \emph{dynamical system} (DS).
We assume all the variables are
dimensionless; furthermore, without
a loss of generality, we assume that all
 monomials in~$f_i(\mathbf{x})
= f_i(\mathbf{x}; n)$ are distinct.
Central to this paper is a special subset of these DSs.

\begin{definition} $($\textbf{Chemical dynamical system}$)$ 
\label{def:CDS}
Assume that $m(\mathbf{x})$ is a monomial in $f_i(\mathbf{x})$
such that $m(\mathbf{x}) \ge 0$ 
when $x_i = 0$ and $x_{j} \ge 0$ for all $j \ne i$.
Then $m(\mathbf{x})$ is said to be a \emph{chemical} monomial.
If $f_i(\mathbf{x})$ contains only chemical monomials,
then $f_i(\mathbf{x})$ is \emph{chemical}.
If for every $i = 1,2,\ldots,N$ 
function $f_i(\mathbf{x})$ is chemical, 
then~$(\ref{eq:DS})$ 
is an $N$-dimensional $n$-degree (mass-action) 
\emph{chemical dynamical system (CDS)}.
\end{definition}

Let $\mathbf{x}(t;\mathbf{x}_0)$ be the solution
of a CDS with initial condition 
$\mathbf{x}(0;\mathbf{x}_0) = \mathbf{x}_0$.
It follows from Definition~\ref{def:CDS} that 
if the initial condition is non-negative, 
$\mathbf{x}_0 \in \mathbb{R}_{\ge}^N$, 
then the solution remains non-negative for all future times, 
$\mathbf{x}(t;\mathbf{x}_0) \in \mathbb{R}_{\ge}^N$
for every $t \ge 0$~\cite{Janos,Feinberg}. 
Furthermore, variables $x_1,x_2,\ldots,x_N$ from CDSs 
can be interpreted 
as (non-negative) concentrations of chemical species,
which react according to a set
of reactions jointly called a \emph{chemical reaction network} (CRN)~\cite{Janos,Feinberg}; 
see Appendix~\ref{app:CRN} for more details.
In this paper, we call the pair consisting of
a CDS and an induced CRN a \emph{chemical system}.

To characterize structural complexity of DSs and CRNs, 
we use the following definition.

\begin{definition} $($\textbf{Structural complexity}$)$ 
\label{def:CDS_complexity}
Let $M_d$ be the total number of monomials of degree $d$
contained in $f_1(\mathbf{x}; n), f_2(\mathbf{x}; n), 
\ldots, f_N(\mathbf{x}; n)$.
Then \emph{DS}~$(\ref{eq:DS})$ is said to be 
an $n$-degree $(M_0 + M_1 + \ldots + M_n,M_2,M_3,\ldots,M_n)$ \emph{DS}. If a \emph{CRN} induced by an $n$-degree \emph{CDS} 
contains exactly $R_d$ reactions of degree $d$,
then it is said to be an $n$-degree
$(R_0 + R_1 + \ldots + R_n,R_2,R_3,\ldots,R_n)$ \emph{CRN}.
\end{definition}

In this paper, we focus on DSs with $N = 3$,
and then let $x = x_1$, $y = x_2$ and $z = x_3$.

\begin{example} 
Consider the three-dimensional quadratic \emph{DS}
\begin{align}
\frac{\mathrm{d} x}{\mathrm{d} t}
& = f_1(y,z) =  2.7 y + z, \nonumber \\
\frac{\mathrm{d} y}{\mathrm{d} t}
& = f_2(x,y) = -x + y^2, \nonumber \\
\frac{\mathrm{d} z}{\mathrm{d} t}
& = f_3(x,y) = x + y,
\label{eq:Sprott_P}
\end{align}
given as system $P$ in~\emph{\cite{Sprott}[Table 1]}.
By \emph{Definition}~$\ref{def:CDS_complexity}$,
since $M_0 = 0$, $M_1 = 5$ and $M_2 = 1$,
$(\ref{eq:Sprott_P})$ is a quadratic 
$(0 + 5 + 1, 1) = (6,1)$ \emph{DS}.
Since $f_2(x,y)$ contains the non-chemical
monomial $-x$, $(\ref{eq:Sprott_P})$
is not a \emph{CDS} by \emph{Definition~\ref{def:CDS}}, 
and hence there is no associated \emph{CRN}.
\end{example}

\begin{example} 
\label{ex:10_3}
Consider the three-dimensional quadratic \emph{DS}
\begin{align}
\frac{\mathrm{d} x}{\mathrm{d} t} 
& = f_1(x,y) = \alpha_1 + \alpha_2 x + \alpha_3 x^2
- \alpha_4 x y, \nonumber \\
\frac{\mathrm{d} y}{\mathrm{d} t} 
& = f_2(x,y,z) = \alpha_2 x - \alpha_5 y + \alpha_6 z, 
\nonumber \\
\frac{\mathrm{d} z}{\mathrm{d} t} 
& = f_3(x,y,z) = \alpha_2 x - \alpha_7 z + \alpha_8 y z, 
\label{eq:CDS_P}
\end{align}
with coefficients $\alpha_1,\alpha_2,\ldots,\alpha_8 > 0$.
Since every monomial with a negative coefficient 
is multiplied by $x$ in $f_1(x,y)$, 
this function is chemical;
the same applies to the second and third equations.
Therefore, by \emph{Definition~\ref{def:CDS}}, 
$(\ref{eq:CDS_P})$ is a \emph{CDS}.
By \emph{Definition}~$\ref{def:CDS_complexity}$,
since $M_0 = 1$, $M_1 = 6$ and $M_2 = 3$,
$(\ref{eq:CDS_P})$ is a quadratic 
$(10,3)$ \emph{CDS}.
This \emph{CDS} is dynamically similar
to~$(\ref{eq:Sprott_P})$, and constructed 
in~\emph{\cite{Chaos_1}[Theorem 4.3]}.
 
Denoting by $X, Y, Z$ the chemical species
with concentrations $x, y, z$, 
an associated \emph{CRN} reads
\begin{align}
\varnothing & \xrightarrow[]{\alpha_1} X, \; \; \; 
X \xrightarrow[]{\alpha_2} 2 X,  \; \; \; 
X \xrightarrow[]{\alpha_2} X + Y,  \; \; \; 
X \xrightarrow[]{\alpha_2} X + Z, \; \; \; 
2 X \xrightarrow[]{\alpha_3} 3 X, \nonumber \\
X + Y & \xrightarrow[]{\alpha_4} Y, \; \; \; 
Y \xrightarrow[]{\alpha_5} \varnothing, \; \; \; 
Z \xrightarrow[]{\alpha_6} Y + Z, \; \; \; 
Z \xrightarrow[]{\alpha_7} \varnothing, \; \; \; 
Y + Z \xrightarrow[]{\alpha_8} Y + 2 Z,
\label{eq:CRN_P_1}
\end{align}
where $\varnothing$ denotes some species
not explicitly modelled. Since 
$R_0 = 1$, $R_1 = 6$ and $R_2 = 3$, 
$(\ref{eq:CRN_P_1})$ is a quadratic
$(10,3)$ \emph{CRN} by 
\emph{Definition}~$\ref{def:CDS_complexity}$. 
This network is known as the \emph{canonical CRN}
induced by~$(\ref{eq:CDS_P})$. 
Since the monomials $x$ are multiplied,
up to sign, by the same coefficients in 
the first, second and third equations, 
the three canonical reactions 
$X \xrightarrow[]{\alpha_2} 2 X$, 
$X \xrightarrow[]{\alpha_2} X + Y$, 
$X \xrightarrow[]{\alpha_2} X + Z$
can be \emph{fused} into a single 
non-canonical reaction 
$X \xrightarrow[]{\alpha_2} 2 X + Y + Z$. 
This fusion leads to a quadratic $(8,3)$ \emph{CRN}
associated to $(\ref{eq:CDS_P})$, given by
\begin{align}
\varnothing & \xrightarrow[]{\alpha_1} X, \; \; \; 
X \xrightarrow[]{\alpha_2} 2 X + Y + Z,  \; \; \; 
2 X \xrightarrow[]{\alpha_3} 3 X, \; \; \; 
X + Y \xrightarrow[]{\alpha_4} Y,\nonumber \\
Y & \xrightarrow[]{\alpha_5} \varnothing, \; \; \; 
Z \xrightarrow[]{\alpha_6} Y + Z, \; \; \; 
Z \xrightarrow[]{\alpha_7} \varnothing, \; \; \; 
Y + Z \xrightarrow[]{\alpha_8} Y + 2 Z.
\label{eq:CRN_P}
\end{align}
See \emph{Appendix~\ref{app:CRN}} for 
more details on how to construct \emph{CRN}s
for any given \emph{CDS}.
\end{example}

\textbf{Chaos}. To discuss chaos,
let us first present a basic definition.

\begin{definition}
\label{def:invariant_set} 
Let $\mathbf{x}(t;\mathbf{x}_0)$ 
be the solution of~$(\ref{eq:DS})$
with $\mathbf{x}(0;\mathbf{x}_0) = \mathbf{x}_0$.
Assume that the following statement holds:
if $\mathbf{x}_0 \in \mathbb{V} \subset \mathbb{R}^N$, 
then $\mathbf{x}(t;\mathbf{x}_0) \in \mathbb{V}$ 
for all $t \in \mathbb{R}$.
Then, $\mathbb{V}$ is said to be an
\emph{invariant set} for \emph{DS}~$(\ref{eq:DS})$.
\end{definition} 

Loosely speaking, a compact invariant set $\mathbb{V}$ 
for DS~(\ref{eq:DS}) is \emph{chaotic} if it has the 
following two properties~\cite{Wiggins,Sprott_book_1}:
(i) trajectories in $\mathbb{V}$ 
are ``sensitive to initial conditions'', and
(ii) $\mathbb{V}$ is ``irreducible'';
$\mathbb{V}$ is said to be 
a \emph{chaotic attractor} if additionally 
(iii) there exists a neighborhood 
$\mathbb{U} \supset \mathbb{V}$ 
such that if $\mathbf{x}(0;\mathbf{x}_0) = \mathbf{x}_0 \in \mathbb{U}$,
then $\mathbf{x}(t;\mathbf{x}_0)$ approaches 
$\mathbb{V}$ as $t \to \infty$. 
Properties (i) and (ii) 
can be defined in a number of non-equivalent 
ways~\cite{Chaos_Def}. For example, (i) can be defined
via existence of a positive 
\emph{Lyapunov exponent} (LE)~\cite{Lyapunov,Sprott_book_1},
which implies that nearby trajectories in $\mathbb{V}$ 
separate exponentially fast.
More broadly, there are also other
definitions of chaos that impose additional properties, 
or relax some of the stated properties~\cite{Chaos_Def};
however, all the definitions 
explicitly or implicitly exclude chaos
from occurring in DSs~(\ref{eq:DS})
with $N \le 2$ or $n = 1$.
In this paper, we use only some 
properties of chaotic sets.
In particular, in Section~\ref{sec:theory},
we present some results that rely
only on compactness of chaotic sets,
and on the fact that they occur in 
DSs~(\ref{eq:DS}) only 
if $N \ge 3$ and $n \ge 2$;
in Section~\ref{sec:examples},
we use positive LEs as 
numerical indicators of chaos.

\section{Theory}
\label{sec:theory}
In this section, we prove some elementary
results about CDSs with chaos. 
We start by establishing some constraints on the sign and 
number of monomials in DSs that have a compact invariant 
set in the non-negative orthant. More specifically, 
to ensure that the invariant set has a significant
positive portion, we also demand that 
at least one of the underlying solutions spends infinite
amount of time in the positive orthant.

\begin{theorem} $($\textbf{Sign theorem}$)$
\label{theorem:sign}
Let $\mathbb{V} \subset \mathbb{R}_{\ge}^N$
be a compact invariant set for \emph{DS}~$(\ref{eq:DS})$.
Assume that there exists 
$\mathbf{x}_0 \in \mathbb{V}$ such that
the solution $\mathbf{x}(t;\mathbf{x}_0)$
of~$(\ref{eq:DS})$ spends infinite amount of time 
in some non-empty compact set $\mathbb{K} \subset
\mathbb{V} \cap \mathbb{R}_{>}^N$.
Then, $(\ref{eq:DS})$ has at least $2 N$
monomials. In particular, $f_i(\mathbf{x})$ has 
at least one monomial with a negative coefficient 
and at least one monomial with a positive coefficient 
for every $i = 1,2,\ldots,N$.
\end{theorem}

\begin{proof}
For contradiction, assume that for some $i$ 
all monomials in $f_i(\mathbf{x})$
have negative (respectively, positive) coefficients. 
By compactness of $\mathbb{K} \subset \mathbb{R}_{>}^N$, 
there exist $m, M > 0$ such that 
for every $\mathbf{x} \in \mathbb{K}$
every component $x_i$ satisfies $m \le x_i \le M$.
Furthermore, there exists $F > 0$ such that 
$f_i(\mathbf{x}) \le - F$
(respectively, $f_i(\mathbf{x}) \ge F$)
for all $\mathbf{x} \in \mathbb{K}$.
Hence, $x_i(t;\mathbf{x}_0)$ 
monotonically decreases (respectively, increases)
during every time-interval during which
$\mathbf{x}(t;\mathbf{x}_0) \in \mathbb{K}$.
Furthermore, since $f_i(\mathbf{x}) \le 0$
(respectively, $f_i(\mathbf{x}) \ge 0$)
for all $\mathbf{x} \in \mathbb{V} \subset \mathbb{R}_{\ge}^N$, 
$x_i(t;\mathbf{x}_0)$ cannot increase
(respectively, decrease) when 
$\mathbf{x}(t;\mathbf{x}_0) \not\in \mathbb{K}$.
By assumption,  
$\mathbf{x}(t;\mathbf{x}_0)$
spends infinite amount of time in $\mathbb{K}$
for $t \in [0,\infty)$, 
so that there exists
$t^* > 0$ such that
$x_i(t^*;\mathbf{x}_0) < m$
(respectively, $x_i(t^*;\mathbf{x}_0) > M$),
and therefore $\mathbf{x}(t;\mathbf{x}_0) 
\not \in \mathbb{K}$ for all $t \ge t^*$, 
leading to contradiction. 
\end{proof}

\noindent
\textbf{Remark}. If $\mathbb{V}$
is in the positive orthant,
$\mathbb{V} \subset \mathbb{R}_{>}^N$, 
then Theorem~\ref{theorem:sign} holds trivially
with $\mathbb{K} = \mathbb{V}$. 

\noindent
\textbf{Remark}. By Theorem~\ref{theorem:sign},
three-dimensional polynomial DSs with chaos in 
$\mathbb{R}_{>}^3$ have at least $6$ monomials.

We now show that three-dimensional 
CDSs with chaos have a monomial  
of particular form and sign.

\begin{theorem} 
\label{theorem:non_linearity}
Assume that~$(\ref{eq:DS})$ 
is a \emph{CDS} with $N = 3$ 
and a chaotic set 
in $\mathbb{R}_{\ge}^3$.
Then, for some $i$ polynomial
$f_i(\mathbf{x})$ contains 
a monomial of the form 
$-\alpha x_1^{\nu_1} x_2^{\nu_2} x_3^{\nu_3}$, 
where $\alpha > 0$, $\nu_i \ge 1$
and $\nu_k \ge 1$ for some $k \ne i$.
\end{theorem}

\begin{proof}
For contradiction, assume that
for every $i = 1,2,3$ polynomial
$f_i(\mathbf{x})$ contains no monomials
of the form $-\alpha x_1^{\nu_1} x_2^{\nu_2} x_3^{\nu_3}$, 
where $\alpha > 0$, $\nu_i \ge 1$
and $\nu_k \ge 1$ for some $k \ne i$.
Then, by Definition~\ref{def:CDS}, 
if a monomial in $f_i(\mathbf{x})$ 
is multiplied by a negative
coefficient, then it takes the form 
$-\alpha x_i^{\nu_i}$ with $\alpha > 0$ and $\nu_i \ge 1$.
The Jacobian matrix for~(\ref{eq:DS}) then has 
non-negative off-diagonal elements in $\mathbb{R}_{\ge}^3$,
i.e. (\ref{eq:DS}) is then \emph{cooperative}~\cite{Hirsch}.
It then follows from~\cite{Hirsch}[Theorem A]
that every compact invariant set for DS~(\ref{eq:DS})
in $\mathbb{R}_{\ge}^3$ also exists for some 
two-dimensional DS with Lipschitz right-hand side;
consequently, (\ref{eq:DS}) cannot
have a chaotic set in $\mathbb{R}_{\ge}^3$.
\end{proof}

\textbf{CDSs with only one quadratic}. 
Let us now further restrict DSs~(\ref{eq:DS}) with $N = 3$, 
by also demanding that $n = 2$ and that there is 
only one quadratic monomial.
Such DSs can be transformed via permutation of the 
dependent variables and time-rescaling into
\begin{align}
\frac{\mathrm{d} x}{\mathrm{d} t} & = 
\alpha_1 + \alpha_2 x + \alpha_3 y + \alpha_4 z
+ m(x,y,z), \nonumber \\
\frac{\mathrm{d} y}{\mathrm{d} t} & = 
\alpha_5 + \alpha_6 x - \alpha_7 y + \alpha_8 z,
\nonumber \\
\frac{\mathrm{d} z}{\mathrm{d} t} & = 
\alpha_9 + \alpha_{10} x + \alpha_{11} y - \alpha_{12} z,
\label{eq:DS_quadratic_1}
\end{align}
where $m(x,y,z) \in \{\pm x^2, 
\pm y^2, \pm x y, \pm y z\}$. 
Chaos is reported in DSs of the 
form~(\ref{eq:DS_quadratic_1}) 
for every choice of $m(x)$; see $(6,1)$ 
Sprott systems F--S~\cite{Sprott}[Table 1].
In contrast, Theorem~\ref{theorem:non_linearity}
implies that such CDSs may possibly have chaos 
for only one particular choice of $m(x,y,z)$. 

\begin{theorem} 
\label{theorem:non_existence}
Assume that~$(\ref{eq:DS_quadratic_1})$
is a \emph{CDS}. If $m(x,y,z) \ne - x y$,
then $(\ref{eq:DS_quadratic_1})$ 
has no chaotic set in $\mathbb{R}_{\ge}^3$.
If $m(x,y,z) = - x y$ and if~$(\ref{eq:DS_quadratic_1})$ has a 
chaotic set in $\mathbb{R}_{>}^3$,
then $\alpha_7 > 0$, $\alpha_{12} > 0$
and $\alpha_7 \alpha_{12} - \alpha_8 \alpha_{11} > 0$. 
\end{theorem}

\begin{proof}
The first statement follows directly from Theorem~\ref{theorem:non_linearity}.
Let us assume that $m(x,y,z) = - x y$
and that a chaotic set 
$\mathbb{V} \subset \mathbb{R}_{>}^3$ exists.
By Definition~\ref{def:CDS}, 
$\alpha_i \ge 0$ for all $i \ne 2, 7, 12$
and, by Theorem~\ref{theorem:sign}, 
$\alpha_7, \alpha_{12} > 0$.
Furthermore, $\alpha_6 > 0$ or $\alpha_{10} > 0$;
otherwise, if $\alpha_{6} = \alpha_{10} = 0$,
then chaos is not possible, since chaos
cannot occur in one- or two-dimensional DSs. 
Assume for contradiction that
$\alpha_7 \alpha_{12} - \alpha_8 \alpha_{11} \le 0$,
which can hold only if $\alpha_8, \alpha_{11} > 0$.  
Let $u(t) \equiv a y(t) + b z(t)$ and choose any
$a, b > 0$ such that $a/b \in 
[\alpha_{12}/\alpha_8, \alpha_{11}/\alpha_7]$. 
Then, by compactness of $\mathbb{V} \subset \mathbb{R}_{>}^3$, 
there exists $G > 0$ such that for every 
$(x,y,z) \in \mathbb{V}$
\begin{align}
\frac{\mathrm{d} u}{\mathrm{d} t} & = 
(a \alpha_5 + b \alpha_9)
+ (a \alpha_6 + b \alpha_{10}) x
+ (-a \alpha_{7} + b \alpha_{11}) y
+ (a \alpha_8 - b \alpha_{12}) z
\ge (a \alpha_6 + b \alpha_{10}) x \ge G.
\end{align}
Hence, $y(t)$ or $z(t)$ then increases 
monotonically, which is not possible 
if $\mathbb{V}$ is both invariant and compact. 
\end{proof}

\section{Examples}
\label{sec:examples}
Let us consider DS~(\ref{eq:DS}) with $N = 3$ and $n = 2$,
which we write as follows:
\begin{align}
\frac{\mathrm{d} x}{\mathrm{d} t} & = \alpha_1 + \alpha_2 x + \alpha_3 y + \alpha_4 z 
+ \alpha_5 x^2 + \alpha_6 y^2 + \alpha_7 z^2 + \alpha_8 x y + \alpha_9 x z + \alpha_{10} y z,
\nonumber \\
\frac{\mathrm{d} y}{\mathrm{d} t} & =  \alpha_{11} + \alpha_{12} x + \alpha_{13} y + \alpha_{14} z + \alpha_{15} x^2 + \alpha_{16} y^2 + \alpha_{17} z^2 + \alpha_{18} x y + \alpha_{19} x z + \alpha_{20} y z,
\nonumber \\
\frac{\mathrm{d} z}{\mathrm{d} t} & = \alpha_{21} + \alpha_{22} x + \alpha_{23} y + \alpha_{24} z + \alpha_{25} x^2 + \alpha_{26} y^2 + \alpha_{27} z^2 + \alpha_{28} x y + \alpha_{29} x z + \alpha_{30} y z.
\label{eq:CDS_general}
\end{align}
Necessary for~(\ref{eq:CDS_general}) to be a CDS
with a chaotic set in $\mathbb{R}_{\ge}^3$
is that the following conditions hold:
\begin{align}
\alpha_{i} & \ge 0 \; \; \textrm{for every}
\; \; i = 1,3,4,6,7,10,11,12,14,15,17,19,21,22,23,25,26,28,
\nonumber \\
\alpha_{i} & < 0 \; \; \textrm{for some}
\; \; i = 8,9,18,20,29,30.
\label{eq:CDS_constraint}
\end{align}
In particular, the first line of inequalities
in~(\ref{eq:CDS_constraint}) 
follows directly from Definition~\ref{def:CDS},
while the second line follows from  Theorem~\ref{theorem:non_linearity}.
Furthermore, by Theorem~\ref{theorem:sign}, 
for a chaotic set to exist in $\mathbb{R}_{\ge}^3$, 
with at least one solution that spends 
infinite amount of time in $\mathbb{R}_{>}^3$,
it is necessary that 
every equation of~(\ref{eq:CDS_general}) has 
at least one negative and one positive coefficient, 
resulting in a total of at least $6$ non-zero coefficients. 

Under these constraints, 
we have performed a computer search for 
chaotic CDSs with only $6$ or $7$ monomials 
by considering a large number of combinations 
of the coefficients $\alpha_1,\alpha_2,\ldots,\alpha_{30}$
and initial conditions, 
and choosing only the cases with a detected solution
whose largest LE is at least $0.001$.
For each chosen case, a secondary refined
computer search was performed with the aim to
nullify as many coefficients as possible, 
while fixing the others to either $\pm 1$, 
a small integer or a rational number
with a small number of decimal places,
all while preserving a positive LE;
such systems are known as ``elegant'' in the literature~\cite{Sprott_book_1,Sprott_book_2}.
Furthermore, if a monomial appeared in multiple equations
of the resulting CDSs, then the dependent variables
have been rescaled with positive numbers
such that these monomials are multiplied,
up to sign, by the same coefficients, with
priority given to quadratic monomials.
This procedure reduces the number of (quadratic) 
reactions in the resulting CRNs;
see Example~\ref{ex:10_3} and Appendix~\ref{app:CRN}.
Finally, if two discovered CDSs
were related via only permutation of the variables,
or if they had the same monomials multiplied by 
sign-matched coefficients (e.g. if they are related
via rescaling of the dependent variables with positive numbers),
or if one of the systems had some additional monomials
(e.g. if it is a perturbation of the other CDS),
then we have reported only one of the systems, 
which has lower monomial number and 
more elegant coefficients.

Using this approach, we have discovered
$20$ three-dimensional quadratic chemical systems, 
named $\textrm{CS}_1$--$\textrm{CS}_{20}$
and presented in Tables~\ref{tab:one}--\ref{tab:two}, 
whose numerically detected chaotic attractors
are shown in the $(x,y)$-space in Figure~\ref{fig:1}.
In the tables, each row describes one chemical system:
the first column contains the name of the system
and structural complexity of its CDS and CRN
(see Definition~\ref{def:CDS_complexity}), 
which are respectively presented in the
second and third column. The fourth column shows
the non-negative equilibria of the CDS in the form 
``$\textrm{type}_{i,j}(x^*,y^*,z^*)$'', 
which is defined as follows:
equilibrium $(x^*,y^*,z^*)$ has $i$ eigenvalues 
with negative, and $j$ eigenvalues with positive, 
real parts, and is of type
``$\textrm{n}$'', ``$\textrm{s}$'', ``$\textrm{sf}$''
or ``$\textrm{nh}$'' denoting a node, saddle, 
saddle-focus and a non-hyperbolic equilibrium, 
respectively; the coordinates $(x^*,y^*,z^*)$
are truncated to two decimal places.
The fifth column displays 
three LEs of the chaotic attractor, 
with the corresponding
Lyapunov (Kaplan–Yorke) dimension (LD)~\cite{Sprott_book_1,Sprott_book_2} 
shown in the sixth column.
These quantities have been computed 
for a solution in the basin of attraction 
and close to the chaotic attractor, 
with the integer initial condition (IC)
displayed in the seventh column.
In particular, the largest LE has been 
computed at time $t = 4 \times 10^7$ 
and truncated to four decimal places. 
We have then assumed existence of a chaotic
set with suitable regularity conditions,
so that the second LE is zero, while the third
one is obtained by subtracting the largest LE 
from the trace of the Jacobian averaged along the solution,
which has been computed numerically;
see~\cite{Sprott_book_1,Sprott_book_2} for more details.
For example, $\textrm{CS}_1$ presented in 
the first row of Table~\ref{tab:one} consists of 
a $(6,4)$ CDS and a $(5,3)$ CRN, 
and has a unique equilibrium $(0,0,0)$,
which is non-hyperbolic with one negative eigenvalue
and no eigenvalues with positive real parts;
furthermore, the three LEs are $0.1597,0,-6.9947$, 
the LD is $2.0228$, and an initial condition 
in the basin of attraction 
of the chaotic attractor 
is $(x_0,y_0,z_0) = (4,2,3)$. 

Eight systems with $6$ monomials are shown 
in Table~\ref{tab:one}, 
while the remaining twelve with $7$ monomials 
in Table~\ref{tab:two}.
In particular, there are
three $(6,4)$, five $(6,5)$, one $(7,2)$, 
three $(7,3)$ and eight $(7,4)$ CDSs.
At the level of CRNs, there is one $(4,3)$ 
(see $\textrm{CS}_{6}$), 
four $(5,2)$ (see $\textrm{CS}_{9}$, 
$\textrm{CS}_{12}$, $\textrm{CS}_{13}$, 
$\textrm{CS}_{16}$),
seven $(5,3)$ (see $\textrm{CS}_{1}$, 
$\textrm{CS}_{2}$, $\textrm{CS}_{3}$, 
$\textrm{CS}_{14}$, $\textrm{CS}_{17}$, 
$\textrm{CS}_{18}$, $\textrm{CS}_{19}$), 
four $(5,4)$ (see $\textrm{CS}_{4}$, 
$\textrm{CS}_{5}$, $\textrm{CS}_{7}$,
$\textrm{CS}_{8}$), 
two $(6,3)$ (see $\textrm{CS}_{10}$,
$\textrm{CS}_{11}$) and two $(7,4)$ 
(see $\textrm{CS}_{15}$, $\textrm{CS}_{20}$) CRNs.
Note that $\textrm{CS}_{16}$ 
is unique in two ways:
its CDS has a zero-degree monomial,
and its CRN only reactions with 
at most two molecules both on the left-
and right-hand side.

$\textrm{CS}_{1}$
has a unique equilibrium, while
$\textrm{CS}_{6}$--$\textrm{CS}_{8}$
have the $x$-axis as a line of equilibria,
which we denote by $(x > 0, 0, 0)$ in Table~\ref{tab:one};
to the best of the authors' knowledge, these three are 
the first reported chaotic CDSs with such property.
The rest of the systems have two or three
equilibria. Furthermore, aside from 
$\textrm{CS}_{1}$, all the CDSs have exactly one
positive equilibrium of saddle-focus type
with a negative eigenvalue.
Note also that the origin is a non-hyperbolic equilibrium
for all of the $6$-monomial systems.
When it comes to the chaotic sets,
all numerically appear to be attractors, 
with LD relatively close to $2$;
furthermore, they all appear to be one-winged.
In addition, the chaotic attractors appear to be
located in the positive orthant, 
except for that of $\textrm{CS}_{15}$
which attains $x = 0$.

\section{Discussion}
\label{sec:discussion}
Chaos in chemical systems 
presented in this paper
has been investigated numerically.
In this context, a natural open problem is:
\emph{Prove (rigorously) that the chemical systems from 
\emph{Tables~\ref{tab:one}--\ref{tab:two}} have chaotic attractors}.
Let us now pose additional open problems, 
to be proved or disproved,
centered on the following question: 
\emph{Are there other chaotic chemical systems
that are as simple or even simpler than} $\textrm{CS}_{1}$--$\textrm{CS}_{20}$?

In particular, chaos is detected in 
a CDS with as low as $2$ quadratic monomials, 
see $\textrm{CS}_{9}$.
In contrast, no chaos has been detected 
in the $(13,1)$ CDS~(\ref{eq:DS_quadratic_1}) under the 
constraints from Theorem~\ref{theorem:non_existence},
suggesting that $1$ quadratic monomial alone
may be insufficient for 
chaos in three-dimensional quadratic CDSs.
\begin{problem}
Prove that for every $n \ge 2$ 
three-dimensional $n$-degree \emph{CDSs}
with only one nonlinear monomial 
have no chaotic set 
in $\mathbb{R}_{\ge}^3$. 
In particular, for $n = 2$, 
consider \emph{DS}~$(\ref{eq:DS_quadratic_1})$
with $m(x,y,z) = -x y$, $\alpha_i \ge 0$ 
for all $i \ne 2,7,12$, 
$\alpha_7, \alpha_{12} > 0$ and $\alpha_7 \alpha_{12}
- \alpha_{8} \alpha_{11} > 0$.
Prove that such \emph{CDS}s have no chaotic set 
in $\mathbb{R}_{\ge}^3$.
\end{problem}

More broadly, $\textrm{CS}_{1}$--$\textrm{CS}_{20}$
all have at least one negative
and one positive quadratic monomial.
One negative quadratic monomial 
is necessary for chaos by Theorem~\ref{theorem:non_linearity}.
Is a positive quadratic monomial also necessary?
\begin{problem}
Prove that for every $n \ge 2$ 
three-dimensional $n$-degree \emph{CDSs}
without at least one nonlinear monomial
with a positive coefficient have no chaotic set 
in $\mathbb{R}_{\ge}^3$.
In particular, for $n = 2$, 
consider \emph{DS}~$(\ref{eq:CDS_general})$
under the constraints~$(\ref{eq:CDS_constraint})$
and $\alpha_i \le 0$ for all $i = 5,\ldots,10,
15,\ldots,20,25,\ldots,30$,
and with at least one negative and one positive
coefficient in each equation.
Prove that such \emph{CDS}s have no chaotic set 
in $\mathbb{R}_{\ge}^3$.
\end{problem}

$\textrm{CS}_{1}$--$\textrm{CS}_{20}$
all have at least one linear monomial.
Is this necessary for chaos?
\begin{problem}
Prove that three-dimensional quadratic \emph{CDSs}
without at least one linear monomial 
have no chaotic set in $\mathbb{R}_{\ge}^3$. 
In particular,
consider \emph{DS}~$(\ref{eq:CDS_general})$
under the constraints~$(\ref{eq:CDS_constraint})$
and $\alpha_i = 0$ for all 
$i = 1,2,3,4,11,12,13,14,21,22,23,24$, 
and with at least one negative and one positive
coefficient in each equation.
Prove that such \emph{CDS}s have no chaotic set 
in $\mathbb{R}_{\ge}^3$.
\end{problem}

The CDSs from Tables~\ref{tab:one}--\ref{tab:two}
have at least one, and at most two, 
equations which factorize, i.e. 
such that all the monomials contain a common factor, 
except for $\textrm{CS}_{17}$
and $\textrm{CS}_{18}$, for which no equation factorizes.
Is chaos possible in $7$-monomial CDSs
with all three equations factorable
(e.g. in the Lotka-Volterra form~\cite{LV})?
\begin{problem}
Consider three-dimensional quadratic 
\emph{CDS}s with at most $7$ monomials such that 
each equation is factorable.
Prove that such \emph{CDS}s have
no chaotic set in $\mathbb{R}_{\ge}^3$.
\end{problem}

$\textrm{CS}_{1}$--$\textrm{CS}_{8}$
have $6$ monomials of which at least
$4$ are quadratic. 
Is chaos possible with fewer quadratics?
\begin{problem}
Prove that no three-dimensional quadratic 
$(6,2)$ or $(6,3)$ \emph{CDS} has a chaotic set 
in $\mathbb{R}_{\ge}^3$.
\end{problem}

$\textrm{CS}_{6}$ has a $(4,3)$ CRN, 
while $\textrm{CS}_{9}$, $\textrm{CS}_{12}$, 
$\textrm{CS}_{13}$ and $\textrm{CS}_{16}$ have $(5,2)$ CRNs;
furthermore, the CRN of $\textrm{CS}_{16}$ has 
at most two molecules on the right-hand side
of the reactions. Are there simpler CRNs with chaos?

\begin{problem}
Prove that no three-dimensional quadratic $(c,d)$ 
\emph{CRN} with $c \le 3$, or $d \le 1$, 
or $(c,d) = (4,2)$,
has a chaotic set in $\mathbb{R}_{\ge}^3$. 
Prove also that no three-dimensional 
quadratic $(c,d)$ \emph{CRN} with $c \le 4$,
and with only reactions with at most two 
molecules on the right-hand side, 
has a chaotic set in $\mathbb{R}_{\ge}^3$.
\end{problem}

In~\cite{Chaos_1}, two chaotic three-dimensional 
CDSs with special properties are constructed:
a quadratic $(11,5)$ CDS with a unique and stable equilibrium,
and a cubic $(11,5,1)$ CDS with a two-wing chaotic attractor. 
In contrast, $\textrm{CS}_{1}$--$\textrm{CS}_{20}$
all have at least one unstable equilibrium in 
$\mathbb{R}_{\ge}^3$, 
and chaotic attractors which appear to have only one wing. 
Can such CDSs display a unique stable equilibrium, 
or two-wing chaos?
\begin{problem}
Consider three-dimensional quadratic 
\emph{CDS}s with at most $7$ monomials
and a chaotic set in $\mathbb{R}_{\ge}^3$.
Prove that such \emph{CDS}s have at least 
one unstable equilibrium in $\mathbb{R}_{\ge}^3$, 
and no two-wing chaotic set in $\mathbb{R}_{\ge}^3$.
\end{problem}

\appendix

\section{Appendix: Chemical reaction networks}
\label{app:CRN}
Every CDS induces a canonical set of chemical reactions~\cite{Janos}.
In what follows, given any $x \in \mathbb{R}$, we
define the sign function as
$\textrm{sign}(x) = -1$ if $x < 0$, 
$\textrm{sign}(x) = 0$ if $x = 0$, and 
$\textrm{sign}(x) = 1$ if $x > 0$.

\begin{definition} $($\textbf{Chemical reaction network}$)$ 
\label{def:CRN}
Assume that~$(\ref{eq:DS})$ is a \emph{CDS}. 
Let $m(\mathbf{x}) = \alpha x_1^{\nu_{1}} 
x_2^{\nu_{2}} \ldots x_N^{\nu_{N}}$
be a monomial from $f_i(\mathbf{x})$,
where $\alpha \in \mathbb{R}$
and $\nu_1,\nu_2,\ldots,\nu_N$ are non-negative integers.
Then, the monomial $m(\mathbf{x})
= \textrm{\emph{sign}}(\alpha) 
|\alpha| x_1^{\nu_{1}} x_2^{\nu_{2}} \ldots x_N^{\nu_{N}}$
induces the \emph{canonical chemical reaction}
\begin{align}
\sum_{k = 1}^N \nu_{k} X_k 
& \xrightarrow[]{|\alpha|} 
\left(\nu_{i} + \textrm{\emph{sign}}(\alpha) \right) X_i
+ \sum_{k = 1, k \ne i}^N \nu_{k} X_k, 
\label{eq:CR}
\end{align}
where $X_i$ denotes the chemical species whose concentration is $x_i$.
The set of all such chemical reactions,
induced by all the monomials in 
the vector field $\mathbf{f}(\mathbf{x})$,
is called the \emph{canonical chemical reaction network} (\emph{CRN}) 
induced by~$(\ref{eq:DS})$. 
\end{definition}
\noindent \textbf{Remark}. 
Terms of the form $0 X_i$ are denoted by $\varnothing$,
and interpreted as some neglected species.

For any given CDS, the induced canonical CRN 
from Definition~\ref{def:CRN} is unique.
However, a given CDS can also induce other, non-canonical, CRNs.

\begin{definition} $($\textbf{Fused reaction}$)$ \label{def:fused}
Consider $M$ canonical reactions with identical left-hand sides:
\begin{align}
\sum_{k = 1}^N \nu_{k} X_k & \xrightarrow[]{|\alpha_j|} 
\left(\nu_{k_j} + \textrm{\emph{sign}}(\alpha_j) \right) X_{k_j}
+ \sum_{k = 1, k \ne k_j}^N \nu_{k} X_k, 
\; \; \; j = 1,2, \ldots, M. \nonumber 
\end{align}
Assume that these reactions also have identical coefficients
above the reaction arrows, 
$|\alpha_1| = |\alpha_2| = \ldots = |\alpha_M| = |\alpha|$.
Then, the corresponding \emph{fused reaction} is given by 
\begin{align}
\sum_{k = 1}^N \nu_{k} X_k & \xrightarrow[]{|\alpha|} 
\left(\nu_{k_1} + \textrm{\emph{sign}}(\alpha_1) \right) X_{k_1}
+
\ldots
+
\left(\nu_{k_M} + \textrm{\emph{sign}}(\alpha_M) \right) X_{k_M}
+ \sum_{\substack{k = 1,\\ k \ne k_1,k_2,\ldots,k_M}}^N \nu_{k} X_k.
\nonumber
\end{align}
Any network obtained by fusing reactions in the canonical \emph{CRN}
is called a non-canonical \emph{CRN}.
\end{definition}
\noindent \textbf{Remark}. 
Fusion is possible if and only if a monomial appears
in multiple equations of~(\ref{eq:DS}) 
and is multiplied, up to sign, by identical coefficients.

\begin{definition} $($\textbf{Degree of reaction}$)$ 
\label{def:degree}
Chemical reaction $\sum_{k = 1}^N \nu_{k} X_k 
\xrightarrow[]{\alpha} \sum_{k = 1}^N \nu_{k}' X_k $
is said to be of \emph{degree $\sum_{k = 1}^N \nu_{k}$}. 
\end{definition}


\begin{thebibliography}{9}
\bibitem{Wiggins} Wiggins, S. 
Introduction to Applied Nonlinear Dynamical Systems and Chaos.
New York: Springer-Verlag, 1990.

\bibitem{Sprott_book_1} Sprott, J. C. 
Chaos and Time-Series Analysis.
Oxford University Press, Oxford, United Kingdom, 2003.

\bibitem{Lorenz} Lorenz, E. N., 1963.
Deterministic nonperiodic flow. 
Journal of Atmospheric Sciences, 20(2): 130--141.

\bibitem{Rossler} R{\"o}ssler, O.~E., 1976.
An equation for continuous chaos.
Physics Letters A, 57 (5): 397--398.

\bibitem{Sprott} Sprott, J. C., 1994.
Some simple chaotic flows.
Phys. Rev. E, 50, R647(R).

\bibitem{Sprott_51} Sprott, J. C., 1997.
Simplest dissipative chaotic flow. 
Phys. Lett., A, 228: 271--274.

\bibitem{Sprott_0_eq} Jafari, S., Sprott, J. C., 
Hashemi Golpayegani, S. M. R., 2013.
Elementary quadratic chaotic flows with no equilibria.
Physics Letters A. 377, 9: 699--702.

\bibitem{Sprott_1_eq} Molaie, M., Jafari, S., Sprott, J. C., 
Golpayegani, M. R. H., 2013. 
Simple chaotic flows with one stable equilibrium.
International Journal of Bifurcation and Chaos, 23(11).

\bibitem{Sprott_Line}
Jafari, S., Sprott, J.C., 2013.  
Simple chaotic flows with a line equilibrium. 
Chaos Solit. Fract.: 5779--84.

\bibitem{QCM} Plesa, T., 2025. 
Mapping dynamical systems into chemical reactions. 
Nonlinear Dynamics 113: 31149--31173.

\bibitem{Janos} \'{E}rdi, P., T\'{o}th, J. 
Mathematical models of chemical reactions. Theory and applications of deterministic and stochastic Models. 
Manchester University Press, Princeton University Press, 1989.

\bibitem{Feinberg} Feinberg, M. 
Lectures on chemical reaction networks. 
Delivered at the Mathematics Research Center, University of Wisconsin, 1979.

\bibitem{Sprott_book_2} Sprott, J. C. 
Elegant Chaos: Algebraically Simple Chaotic Flows 
World Scientific, 2010.

\bibitem{DNA} Soloveichik, D., Seeing G., Winfree E.: 
DNA as a universal substrate for chemical kinetics.
Proceedings of the National Academy of Sciences, 
107(12), 5393--5398 (2010).

\bibitem{LV2} Arneodo, A., Coullet, P., Tresser, C., 1980. Occurence of strange attractors in threedimensional Volterra equations. Phys. Lett. A 79: 259--263.

\bibitem{RosslerW} Willamowski, K. D., R{\"o}ssler O. E., 1980.
Irregular oscillations in a realistic abstract
quadratic mass action system. 
Z Naturforsch, 35a: 317--318.

\bibitem{LV} Hofbauer, J., Sigmund, K., 1998. 
Evolutionary Games and Population Dynamics. 
Cambridge University Press.

\bibitem{Chaos_1} Plesa, T., 2025. 
Chemical systems with chaos.
Available as https://arxiv.org/abs/2511.15554.

\bibitem{Chaos_Def} Brown, R., Chua, L.O., 1996.
Clarifying chaos: Examples and counterexamples.
International Journal of Bifurcation and Chaos 6: 219--249.

\bibitem{Lyapunov} Lyapunov, A. 
Probl\'em G\'eneral de la Stabilit\'e du Mouvement. 
Ann. Math. Stud., 17, Princeton University Press, Princeton, NJ, 1949.

\bibitem{Hirsch} Hirsch, M.W., 1982. 
Systems of differential equations which 
are competitive or cooperative. I: Limit sets. 
SIAM J. Appl. Math. 13: 167--179.

\end{thebibliography}
\end{document}